\documentclass[english,11pt]{article}
\usepackage[latin1]{inputenc}
\usepackage{amsfonts}
\usepackage{amssymb}
\usepackage{amsmath}
\usepackage{amsthm}
\usepackage{enumerate}
\usepackage{bm}

\usepackage[T1]{fontenc}
\usepackage[mathscr]{euscript}

\usepackage[left=2.5cm, right=2.5cm, top=2.5cm, bottom=3cm]{geometry}

\newtheorem{theorem}{Theorem}

\newtheorem{definition}[theorem]{Definition}
\newtheorem{example}[theorem]{Example}

\newtheorem{lemma}[theorem]{Lemma}

\newtheorem{proposition}[theorem]{Proposition}
\newtheorem{remark}[theorem]{Remark}

\let\ds\displaystyle{}

\def\n{N}
\def\t{T}
\def\w{W}
\def\v{V}
\def\u{U}

\def\M{M^4_q(G)}
\def\N{\bm{N}}
\def\T{\bm{T}}
\def\W{\bm{W}}
\def\U{\bm{U}}
\def\V{\bm{V}}

\def\R{\bm{R}}

\def\D{D}
\def\Partial{\partial}

\def\k{\bm{k}}
\def\p{\bm{p}}
\def\q{\bm{q}}
\def\f{\bm{f}}
\def\h{\bm{h}}
\def\l{\bm{l}}
\def\g{\bm{g}}

\def\rk{\mathop{\rm ord}\nolimits}

\begin{document}
\reversemarginpar

\title{Null curve evolution in four-dimensional pseudo-Euclidean spaces}
\author{José del Amor$^{a}$, Ángel Giménez$^{b}$ and Pascual Lucas$^{a}$ \\
$^{a}${\small Departamento de Matemáticas, Universidad de Murcia}\\
{\small Campus del Espinardo, 30100 Murcia, Spain} \\
$^{b}${\small Centro de Investigación Operativa, Universidad Miguel Hernández de Elche}\\
{\small Avda. Universidad s/n, 03202 Elche (Alicante), Spain}\\
}
\date{}

\maketitle

\begin{abstract}
	We define a Lie bracket on a certain set of local vector fields along a null curve in a $4$-dimensional semi-Riemannian space form. This Lie bracket will be employed to study integrability properties of evolution equations for null curves in a pseudo-Euclidean space. In particular, a geometric recursion operator generating infinite many local symmetries for the null localized induction equation is provided.
\end{abstract}


\section{Introduction}
Recently in \cite{del_amor_hamiltonian_2014,musso_hamiltonian_2010} a connection between the local motion of a null curve in $\mathbb{L}^{3}$ and the celebrated KdV equation was given. In \cite{li_motions_2013} the author obtained a connection between a null curve evolution in $\mathbb{L}^{4}$ (to which we refer to as the ``null localized induction equation'' or NLIE), and the Hirota-Satsuma coupled KdV (HS-cKdV) system, which we remind here briefly. Hirota and Satsuma \cite{hirota_soliton_1981} proposed (perhaps up to rescaling) the HS-cKdV system
\begin{equation}\label{eq-HS}
	\begin{aligned}
		u_t&=\frac{1}{2}u_{xxx}+3uu_x-6vv_x, \\
		v_t&=-v_{xxx}-3uv_x,
	\end{aligned}
\end{equation}
which describes the interactions of two long waves with different dispersion relations. Many systematic methods have been employed in the literature to clarify the integrability of the HS-cKdV system: the Lax pair \cite{dodd_integrability_1982,weiss_modified_1985,weiss_sine-gordon_1984}, the Bäcklund transformation method \cite{levi_hierarchy_1983}, the Darboux transformation \cite{leble_darboux_1993,hu_new_2003,hu_new_2008}, the Painlevé analysis \cite{weiss_modified_1985,weiss_sine-gordon_1984}, the search of infinitely many symmetries and conservation laws \cite{hirota_soliton_1981,fuchssteiner_lie_1982,oevel_integrability_1983}, etc.
Fuchssteiner \cite{fuchssteiner_lie_1982} discovered that HS-cKdV system given by \eqref{eq-HS} admits the symplectic and the cosymplectic operators 
\begin{align}
	\label{eq-symplectic-HS}
	J(u,v)&=\begin{pmatrix} \frac{1}{2}D_{x}+uD_{x}^{-1}+D_{x}^{-1}u  & -2D_{x}^{-1}v \\ -2 v D_{x}^{-1}   & -2 D_{x} \end{pmatrix}, \\
	\label{eq-cosymplectic-HS}
	\Theta(u,v)&=\begin{pmatrix}\frac{1}{2} D_x^3+uD_x+D_{x}u & D_{x}v + v D_{x} \\ D_{x}v + vD_{x}  & \frac{1}{2}D_{x}^3+uD_{x}+D_{x}u \end{pmatrix}.
\end{align}
Furthermore, an infinite hierarchy of symmetries for \eqref{eq-HS} was found in \cite{oevel_integrability_1983},
\begin{equation}\label{eq-flows-HS}
	\sigma_{0}=(u_{x},v_{x});\quad\sigma_{1}=(\frac{1}{2}u_{xxx}+3uu_x-6vv_x,-v_{xxx}-3uv_x);\quad \left\{\sigma_{2n}=(\Theta J)^{n}\sigma_{0};\; \sigma_{2n+1}=(\Theta J)^{n}\sigma_{1}\right\}.
\end{equation}
In this case the recursion operator $\Theta J$ is not hereditary. Hence, the bi-Hamiltonian formulation of the HS-cKdV system does not arise from a Hamiltonian pair. 

In this paper, we extend some of the results given in \cite{del_amor_hamiltonian_2014,musso_hamiltonian_2010,li_motions_2013} to a more general background. More specifically, we generalize the Lie algebra structure defined on the local vector fields along null curves from the $3$-dimensional Minkowski space to $4$-dimensional semi-Riemannian space forms $M^4_q(G)$ of index $q=1$ or $q=2$ and curvature $G$. This Lie algebra together with the properties about the HS-cKdV system described above will be used to construct an infinity hierarchy of commuting symmetries for the NLIE equation in a $4$-dimensional pseudo-Euclidean space.

It is interesting to point out that, from a physical point of view, the $4$-dimensional space is a more realistic context than the $3$-dimensional background, the latter very often serving merely as a toy model. Let us recall that relativistic particles models have been described by actions defined on null curves whose Lagrangians are functions of their curvatures \cite{nersessian_massive_1998,nersessian_particle_2000}. These actions were also studied in Minkowski spaces $\mathbb{L}^{3}$ and $\mathbb{L}^{4}$ (see \cite{ferrandez_geometrical_2002,ferrandez_relativistic_2007}), as well as in $3$-dimensional Lorentzian space forms in \cite{gimenez_relativistic_2010}. All these works were addressed to study variational problems on null curve spaces,  and they have shown that the underlying mechanical system is governed by a stationary system of Korteweg-De Vries type.  Furthermore, if $\gamma(\sigma,0)$ is a critical point (the so-called null elastica) for the action $c\int k \, d\sigma$, where $c$ is a constant, then the associated solution $\gamma(\sigma,t)$ to the NLIE starting from $\gamma(\sigma,0)$ is the null elastica evolving by rigid motions in the direction $X=\frac{1}{2}kT+N$, where $X$ is actually the rotational Killing vector field for the null elastica (see \cite{ferrandez_geometrical_2002,gimenez_relativistic_2010}), and it served to determine the benchmark for the evolution equation NLIE in \cite{del_amor_hamiltonian_2014}. This idea was originally explored by Hasimoto between the elastica (an equilibrium shape of an elastic rod) and the ``localized induction equation'' (LIE). As might be expected, relationships with the Korteweg-De Vries evolution systems still arise when the null curve motions in $4$-dimensional backgrounds are considered.

One of the many advantages of having a scalar evolution equation coming from a curve motion is that many aspects of its integrability can be elucidated from the intrinsic geometry of the involved curves (see \cite{beffa_integrable_2002,sanders_integrable_2003}). Conversely, integrability properties of the curvature flow can be employed to determine integrability properties of the curve evolution equation (see \cite{del_amor_hamiltonian_2014,langer_poisson_1991, mansfield_evolution_2006,del_amor_lie_2015}). Despite the numerous well-known connections between curve evolution equations and integrable Hamiltonian systems of PDEs, there is still a lack of understanding about the mechanisms and links among the different frameworks. Our overall aim here is to go further into those concerns.

The rest of this paper is organized as follows. In section 2 we summarize some basic notions about formal variational calculus on which the Hamiltonian theory of nonlinear evolution equations is based. In section 3 we have included some background formulas and results concerning the differential geometry of null curves in a semi-Riemannian space form. In particular, we study the properties of variation vector fields along a null curve in a semi-Riemannian space form as well as the variational formulas for its curvatures. In section 4, the Lie bracket on the set of $\mathcal{P}$-local vector fields locally preserving the causal character along null curves in a $3$-dimensional Lorentzian space given in \cite{del_amor_hamiltonian_2014} is extended to $4$-dimensional semi-Riemannian space forms. This section also includes a discussion about the connection between the geometric variational formulas for curvatures and the Hamiltonian structure for the HS-cKdV system.  The above results will be employed in section 5 to introduce the NLIE equation as a geometric realization of HS-cKdV equations, and to construct a geometric recursion operator generating an infinity hierarchy of commuting symmetries for the NLIE equation.

\section{Preliminaries}\label{preliminares}
In this section we summarize some necessary notions and basic definitions from differential calculus which are relevant to the rest of the paper (see \cite{dorfman_dirac_1993,dickey_soliton_2003,blaszak_multi-hamiltonian_2012} for a very complete treatment of the subject).  
Let $n$ be a positive integer and consider $u_1,u_2,\ldots,u_{n}$ differentiable functions in the real variable $x$. Set
\begin{equation*}
	u_i^{(m)}=\frac{d^mu_i}{dx^m},\quad \text{for } m\in \mathbb{N}, i\in \left\{1,\ldots,n\right\}.
\end{equation*}

Let $\mathcal{P}$ be the algebra of \emph{polynomials} in $u_1,u_2,\ldots,u_{n}$ and their derivatives of arbitrary order, namely, 
\begin{equation*}
	\mathcal{P}=\mathbb{R}[u^{(m)}_i:m\in \mathbb{N},i\in \left\{1,\ldots,n\right\}].
\end{equation*}
We refer to the elements of $\mathcal{P}$ whose constant term vanishes as $\mathcal{P}_0$. Acting on the algebra $\mathcal{P}$ is defined a derivation $\partial$ obeying
\begin{equation*}
	\begin{cases} 
		\partial(f g)=(\partial f)g+f(\partial g), \\
		\partial(u_i^{(m)})=u_i^{(m+1)},
	\end{cases}
\end{equation*}
thereby becoming a differential algebra. 
\begin{remark}
	In a more general setting we can consider $\mathcal{P}$, for example, to be the algebra of \emph{local functions}, i.e. $\mathcal{P}=\bigcup_{j=1}^{\infty}\mathcal{P}_j$, where $\mathcal{P}_j$ is the algebra of locally analytic functions of $u_{1},u_{2},\ldots,u_{n}$ and their derivatives up to order $j$ (see \cite{mikhailov_symmetry_1987,sokolov_symmetries_1988,mikhailov_symmetry_1991,mikhailov_towards_1998}). All the results and formulas established in sections 4 and 5 involving the algebra $\mathcal{P}$ remain valid if the differential algebra of polynomials is replaced by the differential algebra of local functions. Nonetheless, the differential algebra of polynomials is sufficient for our purposes.
\end{remark}
It is customary to take $\partial$ as the total derivative $D_x$ which can be viewed as
\begin{equation*}
	D_x=\sum_{i=1}^{n}\sum_{m\in \mathbb{N}}u_i^{(m+1)} \frac{\partial }{\partial u_i^{(m)}}.
\end{equation*}
In addition to $\partial$, other derivations $\xi$ may also be considered. The action of $\xi$ is determined if we know how $\xi$  acts on the generators of the algebra. Indeed, set 
$$a_{i,m}=\xi u_i^{(m)},$$
then, for any $f\in \mathcal{P}$ we have
$$\xi f=\sum_{i=1}^{n}\sum_{m\in\mathbb{N}} a_{i,m} \frac{\partial f}{\partial u_i^{(m)}}.$$
The space of all derivations on $\mathcal{P}$, denoted by $\mathop{\rm der}\nolimits(\mathcal{P})$, is a Lie algebra with respect to the usual commutator
\begin{equation*}
	[\partial_1,\partial_2]=\partial_1\partial_2-\partial_2\partial_1,\quad \partial_1,\partial_2\in \mathop{\rm der}\nolimits(\mathcal{P}).
\end{equation*}
Derivations commuting with the total derivative have important properties. Among others, if $[\xi,\partial]=0$, we have
\begin{equation*}
	a_{i,m+1}=\xi u_i^{(m+1)}=\xi\partial u_i^{(m)}=\partial\xi u_i^{(m)}=\partial a_{i,m}.
\end{equation*}
Thus 
\begin{equation*}
	\xi f=\sum_{i=1}^{n}\sum_{m\in\mathbb{N}} a_i^{(m)} \frac{\partial f}{\partial u_i^{(m)}},
\end{equation*}
where $a_i=a_{i,0}=\xi u_i$. Let $A=(a_1,\ldots,a_{n})$ be an element of $\mathcal{P}^{n}$ and write
\begin{equation}\label{derivation}
	\partial_A=\sum_{i=1}^{n}\sum_{m\in\mathbb{N}} a_i^{(m)}\frac{\partial }{\partial u_i^{(m)}}.
\end{equation}
The set of derivations $\partial_A$ is a Lie subalgebra of $\mathop{\rm der}\nolimits(\mathcal{P})$, and it induces a Lie algebra on the space $\mathcal{P}^{n}$. Indeed, a direct computation shows that $\partial_A$ are derivations in $\mathcal{P}$ verifying $[\partial_A,\partial_B]=\partial_{[A,B]}$, 
where $[A,B]=\partial_AB-\partial_BA$, i.e.
\begin{equation*}
	[A,B]_j=\sum_{i=1}^{n}\sum_{m\in\mathbb{N}}\left(a_{i}^{(m)}\frac{\partial b_j}{\partial u_i^{(m)}}-b_{i}^{(m)}\frac{\partial a_j}{\partial u_i^{(m)}}\right).
\end{equation*}
This latter commutator can also be expressed  with the aid of Fréchet derivatives as $[A,B]=B'[A]-A'[B]$, where
\begin{equation}\label{frechet-derivative}
		A'[B]=\left.\frac{d}{dt}\right|_{t=0}\partial_{A}(u+tB), \qquad u=(u_{1},u_{2},\ldots,u_{n}).
\end{equation}
We will refer to $\partial_{A}$ as an evolution derivation (or a vector field, provided that no confusion is possible), and the algebra of all evolution derivations will be denoted by  $\mathop{\rm der}\nolimits^*(\mathcal{P})$. Observe that, in particular, if we take $A=u'=(u'_1,\ldots,u'_{n})$ then $\partial=\partial_{u'}$, being $u'=D_x u$.

Set $u^{(m)}=(u_{1}^{(m)},u_{2}^{(m)},\ldots,u_{n}^{(m)})$ for $m\in \mathbb{N}$. Consider an evolution equation of the form
\begin{equation}\label{evolution-equation}
	\frac{\partial u}{\partial t}=F(u,u^{(1)},u^{(2)},\ldots)
\end{equation}
where $F=(f_{1},\ldots,f_{n})$ is an element of $\mathcal{P}^{n}$. An element $S=(s_{1},s_{2},\ldots,s_{n})\in \mathcal{P}^{n}$ is called a symmetry of the evolution equation \eqref{evolution-equation} if and only if  $[F,S]=F'[S]-S'[F]=0$.
Symmetries of integrable equations can often be generated by recursion operators which are linear operators mapping a symmetry to a new symmetry. A linear differential operator $\mathcal{R}:\mathcal{P}^{n}\rightarrow \mathcal{P}^{n}$ is a recursion operator for the evolution equation \eqref{evolution-equation} if it is invariant under $F$, i.e., $L_F \mathcal{R}=0$, where $L_F$ is the Lie derivative acting as $L_{F}A=[F,A]$ for all $A\in \mathcal{P}^{n}$. $\mathcal{R}$ is said to be hereditary if for an arbitrary vector field $F\in \mathcal{P}^{n}$ the relation $L_{\mathcal{R}F}\mathcal{R}=\mathcal{R}L_F \mathcal{R}$ is verified.

\section{Null curve variations in $\M$}
The geometry of null curves is quite different from the non-null ones, so let us review the relevant results, going further into what concerns us most for later work. 

A semi-Riemannian manifold $(M^n_q,g)$ is an $n$-dimensional differentiable manifold $M^n_q$ endowed with a non-degenerate metric tensor $g$ with signature $(n-q,q)$. The metric tensor $g$ will be also denoted by $\left\langle\cdot , \cdot\right\rangle$ and the Levi-Civita connection by $\nabla$. The sectional curvature of a non-degenerate plane generated by $\left\{u,v\right\}$ is $$K(u,v)=\frac{\left\langle R(u,v)u,v \right\rangle}{\left\langle u,u \right\rangle \left\langle v,v \right\rangle-\left\langle u,v \right\rangle^2},$$
where $R$ is the semi-Riemannian curvature tensor given by
\begin{equation}\label{eq-curvature-tensor}
	R(X,Y)Z=-\nabla_X\nabla_YZ+\nabla_Y\nabla_XZ+\nabla_{[X,Y]}Z.
\end{equation}

Semi-Riemannian manifolds with constant sectional curvature are called semi-Riemannian space forms. It is a well-known fact that the curvature tensor $R$ adopts a simple formula in these manifolds:
\begin{equation}\label{eq-constant-curvature}
	R(X,Y)Z=G \left\{\left\langle Z,X \right\rangle Y- \left\langle Z,Y \right\rangle X\right\},
\end{equation}
where $G$ is the constant sectional curvature. When the curvature $G$ vanishes, then $M^n_q$ is called pseudo-Euclidean  space and will be denoted by $\mathbb{R}^{n}_{q}$.

Let $\M$ denote a $4$-dimensional semi-Riemannian space form with index $q=1,2$, background gravitational field $\left\langle , \right\rangle$ and Levi-Civita connection $\nabla$. A tangent vector $v$ is: timelike if $\left\langle v,v \right\rangle<0$; spacelike if $\left\langle v,v  \right\rangle>0$ or $v=0$; null if $\left\langle  v,v \right\rangle=0$. Therefore, a parametrized curve $\gamma:I\rightarrow \M$ is called null if its tangent vector is null at all points in the curve. Fixed a constant $a>0$, we can consider (if $\gamma$ is not a geodesic) the parameter $\sigma_{a}$ given by 
$$\sigma_a(s)=\frac{1}{\sqrt{a}}\int \left\langle \nabla_{\gamma'(s)}\gamma'(s),\nabla_{\gamma'(s)}\gamma'(s)\right\rangle^{\frac{1}{4}}ds,$$ where $s$ is any parameter. When $a=1$ this parameter agrees with the pseudo arc-length parameter $\sigma$ for the null curve. In fact, it is easy to show that $\sigma_a$ is nothing but a linear reparametrization of the pseudo arc-length parameter and it verifies 
$$\left\langle\nabla_{\gamma'(\sigma_a)}\gamma'(\sigma_a),\nabla_{\gamma'(\sigma_a)}\gamma'(\sigma_a)\right\rangle=a^2.$$

Throughout this paper it will be supposed that we have fixed a constant $a$, $\sigma_a$ will be denoted by $\sigma$ and we will also refer to it as the pseudo arc-length parameter. The Cartan frame of a non-geodesic null curve $\gamma:I\to \M$, verifying that $\left\{\gamma'(\sigma),\gamma''(\sigma),\gamma'''(\sigma)\right\}$ is linearly independent for all $\sigma \in I$, is given by $\left\{\t=\gamma'(\sigma),\w_1,\n,\w_2\right\}$, where
\begin{align*}
	\left\langle \t,\t \right\rangle &=\left\langle \n,\n \right\rangle=0,\qquad \left\langle \t,\n \right\rangle =-1,\\
	\left\langle \w_i,\t \right\rangle &=\left\langle \w_i,\n \right\rangle=0,\qquad \left\langle \w_i,\w_i \right\rangle=\varepsilon_i,
\end{align*}
with $i=1,2$. The Cartan equations read
\begin{equation}\label{eq-cartan}
	\begin{aligned}
		\nabla_T T &= a W_1, \\
		\nabla_T W_1 &= -k_{1} T+a \varepsilon_1 N, \\
		\nabla_T N &= -\varepsilon_1 k_{1} W_1+\varepsilon_2 k_{2} W_2, \\
		\nabla_T W_2 &=  k_{2} T, 
	\end{aligned}
\end{equation}
where $\nabla_\t$ denotes the covariant derivative along $\gamma$ and $k_{1}, k_{2}$ are the \emph{curvatures} of the curve. The fundamental theorem for null curves tells us that $k_{1}$ and $ k_{2}$ determine completely the null curve up to semi-Riemannian isometries (see \cite{ferrandez_null_2001}). Even more, if functions $k_{1}$ and $ k_{2}$ are given we can always construct a null curve, pseudo arc-length parametrized, whose curvature functions are precisely $k_{1}$ and $ k_{2}$. Then any local scalar geometrical invariant defined along a null curve can always be expressed as a function of its curvatures and derivatives. A non-geodesic null curve being pseudo arc-length parametrized and admitting a Cartan frame as above is called a \emph{Cartan curve}. The bundle given by $\mathop{\rm span}\nolimits \left\{W_1,W_2\right\}$ is known as the screen bundle of $\gamma$ (see \cite{ferrandez_null_2001}). Projections of the variation vector fields onto the screen bundle will play a leading role in this research.

Let $\gamma$ be a null curve, for the sake of simplicity the letter $\gamma$ will also denote a variation of null curves (null variation) $\gamma=\gamma(s,t):I\times(-\zeta,\zeta)\rightarrow \M$ with $\gamma(s,0)$ the initial null curve. Associated with such a variation is the variation vector field $\v(s)=\v(s,0)$, where $\v=\v(s,t)=\frac{\partial \gamma}{\partial t}(s,t)$. We denote by $\eta$ the differentiable function verifying $\frac{\partial \gamma}{\partial s}(s,t)=\eta(s,t)\t(s,t)$, and by $\frac{D}{\partial t}$ the covariant derivative along the curves $\gamma_s(t)=\gamma(s,t)$. We write $\gamma(\sigma,t), k_{i}(\sigma,t), \v(\sigma,t),$ etc., for the corresponding objects in the pseudo arc-length parametrization.

\begin{definition}
	Let $\mathfrak{X}(\gamma)$ be the set of smooth vector fields along $\gamma$. We say that $\v\in\mathfrak{X}(\gamma)$  locally preserves the causal character if $\left\langle \nabla_\t \v, \t \right\rangle=0$. We also say that $\v$ locally preserves the pseudo arc-length parameter along $\gamma$ if $\eta(s, t)$ 
	satisfies 
	$\left.\frac{\partial \eta}{\partial t}\right|_{\begin{subarray}{1} t=0\end{subarray}}=0$.
\end{definition}

The following properties for null variations can be found in \cite{ferrandez_relativistic_2007} when $a=\varepsilon_{1}=\varepsilon_{2}=1$, but they can be easily adapted to the general situation. 

\begin{lemma}\label{lema-variation1}
	If $\gamma$ is a null variation, then its variation vector field $\v$ verifies
	\begin{equation}\label{eq-variation1}
		\left\langle \nabla_\t \v,\t  \right\rangle=0; \qquad 	\left.\frac{\partial\eta}{\partial t}\right|_{t=0} = -\frac{1}{2a}\rho_V\eta, \qquad [V,T] = \frac{1}{2a}\rho_V T,
	\end{equation}	
	where $\rho_V=-\varepsilon_1\left\langle \nabla_\t^2\v,\w_1  \right\rangle$.
\end{lemma}

Thus we obtain that $\v$ locally preserves the causal character and,  moreover, $\v$ locally preserves the pseudo arc-length parameter if and only if $\rho_\v=0$, which in such a case also entails commutation of $\t$ and $\v$. We define some functions that will play a key role in the rest of the paper, namely, given a vector field $V\in\mathfrak{X}(\gamma)$ we consider the following projections of $V$ and $\nabla_TV$ on the screen bundle given by
\begin{equation}\label{eq-notacion}
		h_V :=\varepsilon_1\langle V,W_1\rangle, \qquad  l_{V} :=\varepsilon_2\langle V,W_2\rangle, \qquad 
		\varphi_{V} :=\varepsilon_1\langle\nabla_T V,W_1\rangle, \qquad  \psi_{V} :=\varepsilon_2\langle\nabla_T V,W_2\rangle.
\end{equation} 

\begin{lemma}\label{lema-variation2} 
	With the above notation, the following assertions hold: 
	\begin{enumerate}[(a)]
		\item $\ds\dfrac{DT}{\partial t}\bigg|_{t = 0}  = -\alpha_V T +\varphi_V W_1 +  \psi_V W_2 $;
		\item $\ds\dfrac{DW_1}{\partial t}\bigg|_{t = 0}  = -\beta_V T +\varepsilon_1\varphi_V N + \dfrac{1}{a}\psi_V' W_2 $;
		\item $\ds\dfrac{DN}{\partial t}\bigg|_{t = 0}  = -\varepsilon_1\beta_V W_1 +\alpha_V N + \dfrac{\varepsilon_1}{a}\delta_V W_2 $;
		\item $\ds\dfrac{DW_2}{\partial t}\bigg|_{t = 0}  = \dfrac{\varepsilon_1\varepsilon_2}{a}\delta_V T -\dfrac{\varepsilon_1\varepsilon_2}{a}\psi_V' W_1 + \varepsilon_2\psi_V N $;
		\item $\ds\dfrac{\partial k_{1}}{\partial t}\bigg|_{t = 0}  = 
		\dfrac{1}{a^2}\varphi_V ''' +\dfrac{1}{a}\left[\left(k_{1}\varphi_V\right)'+k_{1}\varphi_V '\right] 
		-\dfrac{1}{a}\left[\left( k_{2}\psi_V\right)'+ k_{2}\psi_V '\right] 
		+\dfrac{1}{a}\left[\dfrac{1}{2a}\rho_V'' +k_{1}\rho_V -2Gg_V'\right]$;
		\item $\ds\dfrac{\partial  k_{2}}{\partial t}\bigg|_{t = 0}  = 
		\dfrac{\varepsilon_1\varepsilon_2}{a^2}\psi_V ''' +\dfrac{\varepsilon_1\varepsilon_2}{a}\left[\left(k_{1}\psi_V\right)'+k_{1}\psi_V '\right] 
		+\dfrac{1}{a}\left[\left( k_{2}\varphi_V\right)'+ k_{2}\varphi_V '\right] 
		+\dfrac{1}{a} k_{2}\rho_V -\varepsilon_2 Gl_V$;
	\end{enumerate}
	where  
	\begin{equation}\label{eq-notacion1}
		\alpha_V = \dfrac{1}{a}\left(\varphi_V ' +\dfrac{1}{2}\rho_V\right), \quad \beta_V = \dfrac{1}{a}\left(\alpha_V '+ k_{1}\varphi_V - k_{2}\psi_V -Gg_V\right), \quad \delta_V = \dfrac{1}{a}\psi_V '' + k_{1}\psi_V + \varepsilon_1\varepsilon_2 k_{2}\varphi_V. 
	\end{equation}
\end{lemma}
\begin{proof}
	Set $\nabla_V =\frac{D}{\partial t}$ the covariant derivative. From equation \eqref{eq-variation1} we obtain  
	\begin{equation}\label{A}
		\begin{aligned}
				\nabla_V T &= -\left\langle \nabla_V T,N\right\rangle T+\varepsilon_1 \left\langle\nabla_V T,W_1 \right\rangle W_1+\varepsilon_2 \left\langle \nabla_V T,W_2  \right\rangle W_2 \\
						&= -\left\langle \nabla_T V,N  \right\rangle T+\dfrac{1}{2a}\rho_V T+ \varphi_V W_1+ \psi_V W_2 \\
						&= -\dfrac{1}{a}\left(\varphi_V'+\dfrac{1}{2}\rho_V\right) T + \varphi_V W_1+ \psi_V W_2 \\
						&= - \alpha_V T +  \varphi_V W_1 +  \psi_V W_2,
		\end{aligned}
	\end{equation}
	where it has been used $\langle\nabla_T V,N\rangle =\frac{1}{a}\left(\varphi_V ' +\rho_V\right)$. Now, taking into account the formulas \eqref{eq-curvature-tensor}, \eqref{eq-constant-curvature} and \eqref{A}  we have that
	\begin{equation}\label{B}
		\begin{aligned}
				\nabla_V W_1 &= \dfrac{1}{a}\nabla_V\nabla_T T =\dfrac{1}{a}\left(\nabla_T\nabla_V T + \nabla_{[V,T]} T - R(V,T)T\right) \\
						&= \dfrac{1}{a}\left[\nabla_T \left(-\alpha_V T + \varphi_V W_1 + \psi_VW_2\right) +\dfrac{1}{2a}\rho_V\nabla_T T +Gg_VT\right] \\
						&= \dfrac{1}{a}\left[\left(-\alpha_V '-k_{1}\varphi_V +   k_{2}\psi +Gg_V \right) T +a\varepsilon_1\varphi_V N +\psi_V 'W_2\right] \\
						&= -\beta_V T + \varepsilon_1\varphi_V N + \dfrac{1}{a}\psi_V' W_2.
		\end{aligned}
	\end{equation}	
	Considering again \eqref{eq-curvature-tensor}, \eqref{eq-constant-curvature}, \eqref{A} and \eqref{B} we deduce
	\begin{equation}
		\begin{aligned}
				a\varepsilon_1 \nabla_V N &= \nabla_V\nabla_T W_1 +\nabla_V (k_{1} T) \\
				& =\nabla_T\nabla_V W_1 + \nabla_{[V,T]} W_1 - R(V,T)W_1 + V(k_{1})T+k_{1}\nabla_V T \\
						&= \left(-\beta_V '+\dfrac{1}{a} k_{2}\psi_V'-\dfrac{1}{2a}k_{1}\rho_V- k_{1}\alpha_V+V(k_{1})+\dfrac{G}{a}g_V '\right) T \\
						&\quad -a\beta _V W_1+a\varepsilon_1\alpha_V N+ \delta_VW_2. 
		\end{aligned}
	\end{equation}	
	Since $\left\langle \nabla_V N,N \right\rangle=0$, the tangent component of $\nabla_V N$ vanishes, and the expression for $V(k_{1})$ becomes
		\begin{equation}\label{vk}
		\begin{aligned}
				V(k_{1})  &= 
		\dfrac{1}{a^2}\varphi_V ''' +\dfrac{1}{a}\left[\left(k_{1}\varphi_V\right)'+k_{1}\varphi_V '\right] 
		-\dfrac{1}{a}\left[\left( k_{2}\psi_V\right)'+ k_{2}\psi_V '\right] 
		+\dfrac{1}{a}\left[\dfrac{1}{2a}\rho_V'' +k_{1}\rho_V -2Gg_V'\right].
		\end{aligned}
	\end{equation}
	As a consequence the vector field $\nabla_V N$ boils down to
	\begin{equation}\label{C}
		\begin{aligned}
			\nabla_V N &= -\varepsilon_1 \beta_V W_1 +  \alpha_V N + \dfrac{\varepsilon_1 }{a}\delta_V W_2.	
		\end{aligned}
	\end{equation}
	Finally, a similar computation leads to
	\begin{equation}
		\begin{aligned}
				\varepsilon_2 k_{2} \nabla_V W_2 &= \nabla_V\nabla_T N +\varepsilon_1 V(k_{1})W_1 + \varepsilon_1 k_{1}\nabla_V (W_1) -\varepsilon_2 V( k_{2})W_2  \\
						&= \dfrac{\varepsilon_1 }{a} k_{2}\delta_V T + \left(\varepsilon_1 V(k_{1})-\varepsilon_1 \beta_V' -\varepsilon_1\frac{k_{1}}{a}\left(\varphi_V ' + \rho_V\right)+\dfrac{\varepsilon_1}{a}Gg_V '\right) W_1  \\
						&\quad + \left(\alpha_V ' -a\beta_V+ k_{1}\varphi_V -Gg_V \right)N \\
						& \hspace{0.3cm}+ \left(\dfrac{\varepsilon_2}{a} k_{2}\left(\varphi_V ' + \rho_V\right)+\dfrac{\varepsilon_1 }{a}\delta_V ' + \dfrac{\varepsilon_1 }{a}k_{1}\psi_V ' -\varepsilon_2 V( k_{2})-Gl_V\right) W_2.
		\end{aligned}
	\end{equation}	
	In the same way, since the component of $W_2$ in $\nabla_V W_2$ vanishes, we deduce
	\begin{equation}\label{vt}
		\begin{aligned}
				V( k_{2}) &= \dfrac{\varepsilon_1\varepsilon_2}{a^2}\psi_V ''' +\dfrac{\varepsilon_1\varepsilon_2}{a}\left[\left(k_{1}\psi_V\right)'+k_{1}\psi_V '\right] 
		+\dfrac{1}{a}\left[\left( k_{2}\varphi_V\right)'+ k_{2}\varphi_V '\right] 
		+\dfrac{1}{a} k_{2}\rho_V -\varepsilon_2 Gl_V
		\end{aligned}
	\end{equation}
	and
	\begin{equation}\label{D}
		\begin{aligned}
			\nabla_V W_2 &= \dfrac{\varepsilon_1\varepsilon_2}{a}\delta_V T -\dfrac{\varepsilon_1\varepsilon_2}{a}\psi_V' W_1 + \varepsilon_2\psi_V N.
		\end{aligned}
	\end{equation}
\end{proof}

Consider $\Lambda$ the space of pseudo arc-length parametrized null curves in $\M$. For $\gamma\in \Lambda$, it is easy to see that $T_\gamma\Lambda$ is the set of all vector fields associated with variations of pseudo arc-length parametrized null curves starting from $\gamma$. It is clear that a vector field in $T_\gamma\Lambda$ locally preserves the causal character and the pseudo arc-length parameter. The converse can also be proved applying a similar procedure as in \cite{musso_hamiltonian_2010}.

\begin{proposition}\label{prop-symplectic}
	A vector field $V$ along $\gamma\in\Lambda$ is tangent to $\Lambda$ if and only if it locally preserves the causal character and the pseudo arc-length parameter, that is,
	\begin{equation}\label{eq-tlambda}
		T_{\gamma}\Lambda = \{V\in\mathfrak{X}(\gamma): \langle\nabla_T V,T\rangle = \langle\nabla_T^2V,W_1\rangle =0\}.
	\end{equation} 
	Consequently, if $\v$ is expressed by $V=f_V T+h_V W_1 + g_V N + l_V W_2$ where $f_V, h_V , g_V$ and $l_V$  are smooth functions, then $\v\in T_\gamma\Lambda$ if and only if  
	\begin{equation}\label{eq-compatibilidad} 
		f_{V}=-\frac{1}{2a}\left[h'_{V}+a k_{1}D_{\sigma}^{-1}\left(h_{V}\right)-aD_{\sigma}^{-1}\left(k_{1}h_{V}- k_{2} l_{V}\right)\right], \qquad
		g_{V}=-\varepsilon_{1}aD_{\sigma}^{-1}\left(h_{V}\right),
	\end{equation}
	where $D_\sigma^{-1}$ is a formal indefinite $\sigma$-integral. Furthermore, 
	\begin{equation}\label{eq-symplectic}
		\begin{pmatrix} \varphi_V  \\ \psi_V \\ \end{pmatrix}=
		    \begin{pmatrix} \frac{a}{2} \left(\frac{1}{a}D_{\sigma}+k_{1}D_{\sigma}^{-1}+D_{\sigma}^{-1} k_{1}\right) &  &  & -\frac{a}{2}D_{\sigma}^{-1} k_{2} \\ -\varepsilon_1 \varepsilon_2 a k_{2} D_{\sigma}^{-1} &  &  & D_{\sigma} \\ \end{pmatrix}	\cdot     
		    \begin{pmatrix} h_V  \\ l_V \\ \end{pmatrix}             
	\end{equation}	
\end{proposition}
\begin{proof}
	For a generic vector field $V$ we obtain:
	 \begin{align} 
		&\begin{aligned}
		\nabla_T V & =  \left(f_V '-k_{1} h_V + k_{2} l_V\right)T + \left(af_V +h_V ' -\varepsilon_1 k_{1} g_V\right)W_1  \\
			&\quad + \left(\varepsilon_1 ah_V +g_V '\right)N + \left(l_V ' +\varepsilon_2 k_{2} g_V\right)W_2;  
			\label{eq1-derivadas}
		\end{aligned} \\
		&\begin{aligned}
			\nabla_{T}^2 V &= \left[\left(f_V '-k_{1} h_V + k_{2} l_V\right)' -k_{1}\left(af_V+h_V '-\varepsilon_1 k_{1} g_V\right)+ k_{2}\left(l_V ' +\varepsilon_2 k_{2} g_V\right)\right]T \\
				&\quad  +  \left[a\left(f_V '-k_{1} h_V + k_{2} l_V\right) +\left(af_V+h_V '-\varepsilon_1 k_{1} g_V\right)' -\varepsilon_1 k_{1}\left(\varepsilon_1 ah_V + g_V '\right)\right]W_1 \\
				&\quad +  \left[a\varepsilon_1\left(af_V + h_V ' -\varepsilon_1 k_{1} g_V\right) +\left(\varepsilon_1 ah_V + g_V '\right)' \right]N  \\
				&\quad +  \left[\varepsilon_2  k_{2}\left(\varepsilon_1 ah_V + g_V ' \right) +\left(l_V ' + \varepsilon_2 k_{2} g_V \right)' \right] W_2. 
				\label{eq2-derivadas}
		\end{aligned}
	\end{align} 
	If $V\in T_{\gamma}\Lambda$, Lemma \ref{lema-variation1} implies that 
	\begin{equation}\label{eq-compatibilidad2}
			\langle\nabla_T V,T\rangle = \langle\nabla_T^2V,W_1\rangle =0.
	\end{equation}
	In such a case, by using equations \eqref{eq1-derivadas} and \eqref{eq2-derivadas} we deduce that  
	$$\varepsilon_1 ah_V +g'_V=0 \quad\text{and}\quad 2a f_V'+h_{V}''-\varepsilon_{1}(k_{1}g_{V})'-a k_{1} h_V+a k_{2} l_V=0.$$
	Last equations easily give rise to \eqref{eq-compatibilidad}. Expression \eqref{eq-compatibilidad} becomes
	$$	f_V = -\dfrac{1}{2a}h_V ' - \dfrac{1}{2}k_{1}D_{\sigma}^{-1}\left(h_V\right) + \dfrac{1}{2}D_{\sigma}^{-1}\left(k_{1}h_V\right)-\dfrac{1}{2}D_{\sigma}^{-1}( k_{2} l_V),$$
	and the following holds
	\begin{equation}\label{formulas-compatibilidad}
			\varphi_V =\varepsilon_1 \langle \nabla_T V, W_1\rangle = af_V+h_V ' -\varepsilon_1k_{1}g_V, \qquad 
			\psi_V =\varepsilon_2 \langle \nabla_T V, W_2 \rangle = l_V ' +\varepsilon_2 k_{2} g_V.
	\end{equation}
	Replacing  $f_V$ into \eqref{formulas-compatibilidad} and rearranging terms, we easily obtain \eqref{eq-symplectic}. Conversely, if $V$ is a vector field verifying \eqref{eq-compatibilidad2}, then it arises from an infinitesimal variation of null curves.  Indeed, according to the Cartan equations \eqref{eq-cartan}, Lemma \ref{lema-variation2} and formulas \eqref{eq-compatibilidad} and \eqref{formulas-compatibilidad}, we consider the matrices
	\begin{equation*}
		K_{\gamma}=\begin{pmatrix}0&0&0&G&0\\1&0&-k_{1}&0&k_{2}\\0&a&0&-\varepsilon_{1}k_{1}&0\\0&0&\varepsilon_{1}a&0&0\\0&0&0&\varepsilon_{2}k_{2}&0\end{pmatrix},\quad 
		P=\begin{pmatrix}0&G g_{\v}&-\varepsilon_{1}G h_{\v}&G f_{\v}&-\varepsilon_{2}G l_{\v} \\ f_{\v}&-\alpha_{\v}&-\beta_{\v} & 0 & \frac{1}{a}\varepsilon_{1}\varepsilon_{2}\delta_{\v} \\ h_{\v} & \varphi_{\v} & 0 & -\varepsilon_{1}\beta_{\v} & -\frac{\varepsilon_{1}\varepsilon_{2}}{a}\psi'_{\v} \\ g_{\v} & 0 & \varepsilon_{1}\varphi_{\v} & \alpha_{\v} & \varepsilon_{2}\psi_{\v} \\ l_{\v} & \psi_{\v} & \frac{1}{a}\psi'_{\v} & \frac{\varepsilon_{1}}{a}\delta_{\v} & 0 \end{pmatrix},
	\end{equation*}
	verifying
	\begin{equation*}
		\frac{\partial F}{\partial \sigma}=F\cdot K_{\gamma}, \qquad \frac{\partial P}{\partial \sigma}=C-[K_{\gamma},P],
	\end{equation*}
	where
	\begin{equation*}
		F=\begin{pmatrix}\gamma & T & W_{1} & N & W_{2}\end{pmatrix}, \qquad C=\begin{pmatrix}0&0&0&0&0\\0&0&-V(k_{1})&0&V(k_{2})\\0&0&0&-\varepsilon_{1}V(k_{1})&0\\0&0&0&0&0\\0&0&0&\varepsilon_{2}V(k_{2})&0\end{pmatrix}.
	\end{equation*}
	Following the same procedure as described in Lemma 1 of \cite{musso_hamiltonian_2010} we can construct a null curve variation of $\gamma$ whose variation vector field is $\v$.
\end{proof}
From Proposition \ref{prop-symplectic}, a tangent vector field  $V\in T_{\gamma}\Lambda$ and its covariant derivative $\nabla_TV$ are expressed by
\begin{equation}\label{formula-DV} 
	\begin{aligned}
		V &= -\frac{1}{2a}\left[h'_{V}+ak_{1}D_{\sigma}^{-1}\left(h_{V}\right) -aD_{\sigma}^{-1}\left(k_{1}h_{V}- k_{2} l_{V}\right)\right] T \\
		&\quad +h_{V} W_1 -\varepsilon_{1}aD_{\sigma}^{-1}\left(h_{V}\right) N + l_V W_2, \\
		\nabla_T V &= -\frac{1}{a}\varphi_V' T + \varphi_V W_1+\psi_V W_2.
	\end{aligned}
\end{equation}

\begin{remark}\label{constants-integration}
	Observe that a tangent vector field $V\in T_{\gamma}\Lambda$ is completely determined by the differential functions $h_V$ and $l_V$ and two constants, since the operator $D_{\sigma}^{-1}$ is used twice; once for obtaining $g_{V}$ from $h_{V}$ and once more for obtaining $f_{V}$ from $h_{V}$ and $l_{V}$. Therefore, given two differential functions $h_V$ and $l_V$ and two constants, we can construct a vector field locally preserving pseudo arc-length parameter along $\gamma$ whose projections on the screen bundle are precisely $h_V$ and $l_V$. Both constants could be determined or related if constraints on null curve variation are considered, but for our algebraic purposes, we will consider generic constants.
\end{remark}

\section{A Lie algebra structure on local vector fields}
Our objective now is to define a Lie algebra structure on the set of local vector fields which locally preserve the causal character. To this end, we need first to set up the spaces in which we are going to work. Let $k_{1}$ and $ k_{2}$ be smooth functions defined on an interval $I$ and set $\mathcal{P}$ the real algebra of polynomials in  $k_{1}$, $ k_{2}$ and their derivatives of arbitrary order, i.e.,
$$\mathcal{P}=\mathbb{R}\left[ k_{1}^{(m)},  k_{2}^{(n)} : (m,n)\in\mathbb{N}^{2}\right],$$ 
where $ k_{1}^{(m)}= k_{1}^{(m)}(\sigma)$ and $ k_{2}^{(n)}= k_{2}^{(n)}(\sigma)$.{}

Let $\gamma : I\rightarrow \M$ be a null curve with curvatures $k_{1}$ and $k_{2}$, and consider the set of vector fields along $\gamma$ whose components are polynomial functions
\begin{equation*}
	\mathfrak{X}_{\mathcal{P}}(\gamma)=\{V=f_V T + h_V W_1 + g_V N + l_V W_2 \in \mathfrak{X}(\gamma) : f_V , h_V , g_V, l_V \in\mathcal{P}\}.
\end{equation*}
An element of $\mathfrak{X}_{\mathcal{P}}(\gamma)$ will be called a $\mathcal{P}$-local vector field along $\gamma$. The set of $\mathcal{P}$-local vector fields (locally preserving the causal character) will be denoted by
\begin{equation*}
	\mathfrak{X}_{\mathcal{P}}^{*}(\gamma)=\{V=f_V T + h_V W_1 + g_V N + l_V W_2 \in \mathfrak{X}_{\mathcal{P}}(\gamma) : g_{V}=-\varepsilon_{1}aD_{\sigma}^{-1}\left(h_{V}\right)\},
\end{equation*}
and within it, the $\mathcal{P}$-local variation vector fields locally preserving pseudo arc-length parameter are described as
\begin{equation*}
	T_{\mathcal{P},\gamma}\Lambda=T_{\gamma}(\Lambda)\cap \mathfrak{X}_{\mathcal{P}}^{*}(\gamma)=\{V\in \mathfrak{X}_{\mathcal{P}}^{*}(\gamma) : f_{V}=-\frac{1}{2a}\left[h'_{V}+a k_{1}D_{\sigma}^{-1}\left(h_{V}\right)-a D_{\sigma}^{-1}\left(k_{1}h_{V}- k_{2} l_{V}\right)\right] \}.
\end{equation*}
In this context, from Proposition \ref{prop-symplectic} and taking into account Remark \ref{constants-integration}, we can explicitly calculate the $\mathcal{P}$-local pseudo arc-length preserving variation vector fields by means of its constants of integration.

\begin{proposition}\label{prop-constants}
	Let $V$ be a vector field in $\mathfrak{X}_{\mathcal{P}}(\gamma)$, then $V\in T_{\mathcal{P},\gamma}(\Lambda)$ if and only if it is fulfilled that
	\begin{equation*}
		\begin{aligned}
			g_{V}=-\varepsilon_{1}a\partial_{\sigma}^{-1}\left(h_{V}\right)+c_{1},\quad f_{V}=-\frac{1}{2a}\left[h'_{V}+a k_{1}\partial_{\sigma}^{-1}\left(h_{V}\right)-a\partial_{\sigma}^{-1}\left(k_{1}h_{V}- k_{2} l_{V}\right)-\varepsilon_{1}c_{1}k_{1}\right]+c_{2},
		\end{aligned}
	\end{equation*}
	where $c_{1},c_{2}$ are constants, and $\partial_{\sigma}^{-1}$ is the anti-derivative operator verifying that $\partial_{\sigma}^{-1}\circ \partial_{\sigma}=I$ when acting on $\mathcal{P}_0$.
\end{proposition}
Consequently, let us consider the set 
\begin{equation*}
	\mathcal{Q}=\left\{(h,l)\in \mathcal{P}_{0}^{2}: \exists (p,q)\in \mathcal{P}_{0}^{2} \text{ such that } (p',q')=(h,k_{1}h-k_{2}l) \right\}.
\end{equation*}
 Given a pair of functions $(h,l)\in \mathcal{Q}$ and two constants $c_{1}$ and $c_{2}$, we will denote by $\mathcal{X}(h,l)$ the $\mathcal{P}$-local pseudo arc-length preserving variation vector field
\begin{equation}
	\begin{aligned}
		\mathcal{X}	(h,l)&=\left(-\frac{1}{2a}\left[h'+ak_{1}\partial_{\sigma}^{-1}\left(h\right)-a\partial_{\sigma}^{-1}\left(k_{1}h- k_{2} l\right)-\varepsilon_{1}c_{1}k_{1}\right]+c_{2}\right)T \\
		&\quad +h W_{1}+\left(-\varepsilon_{1}a\partial_{\sigma}^{-1}\left(h\right)+c_{1}\right)N+lW_{2}.
	\end{aligned}
\end{equation} 

\begin{example}\label{example-constants}
	Consider the pair of functions $(h,l)=(0,0)$, then $\mathcal{X}(0,0)=\left(\frac{\varepsilon_{1}c_{1}}{2a}k_{1}+c_{2}\right)T+c_{1}N$. In particular, if we take $(c_{1},c_{2})=(0,b)$ and $(c_{1},c_{2})=(-2\varepsilon_{1}a^{2}c,0)$, where $b$ and $c$ are constants, one obtains the vector fields
	\begin{equation*}
		V_{0}=b T,\quad V_{1}=-a c k_{1} T-2\varepsilon_{1}a^{2}c N.
	\end{equation*}
	The vector fields $V_{0}$ and $V_{1}$ will be the starting point of the commuting hierarchy of symmetries in section \ref{hamiltonian-structure}. 
\end{example}

Note that to introduce the concept of symmetry (and so a recursion operator) and furnish the phase space of null curve motions with a formal variational calculus in section \ref{hamiltonian-structure}, an appropriate Lie bracket on the set of local vector fields should be defined. To this end, we first introduce a convenient derivation on both the differential algebra and the local vector fields along a null curve. Motivated by \cite{yasui_differential_1998} and bearing in mind Lemma \ref{lema-variation2}, given $\v\in \mathfrak{X}^*_{\mathcal{P}}(\gamma)$, we denote by $D_\v:\mathfrak{X}_{\mathcal{P}}(\gamma)\rightarrow \mathfrak{X}_{\mathcal{P}}(\gamma)$ the unique tensor derivation fulfilling:
	\begin{align}
			V(f') &= V(f)' +\dfrac{1}{2a}\rho_V f' \quad \mbox{for all} \quad f\in\mathcal{P};
			\label{eq-derivation1}
			\\
			\begin{pmatrix} V(k_{1})  \\ V( k_{2}) \end{pmatrix}  &=  \begin{pmatrix}
			\frac{1}{a}\left(\frac{1}{a}D_{\sigma}^3+k_{1}D_{\sigma}+D_{\sigma}k_{1}\right) & -\dfrac{1}{a}\left(D_{\sigma} k_{2} +  k_{2} D_{\sigma} \right) \\
			\frac{1}{a}\left(D_{\sigma} k_{2} +  k_{2} D_{\sigma} \right) & \dfrac{\varepsilon_1 \varepsilon_2}{a}\left(\frac{1}{a}D_{\sigma}^3+k_{1}D_{\sigma}+D_{\sigma}k_{1}\right) \\
			\end{pmatrix} \cdot \begin{pmatrix}	\varphi_V  \\	\psi_V \\  \end{pmatrix} \label{eq-cosymplectic} \\
			&\quad +\begin{pmatrix}	\frac{1}{a}\left(\frac{1}{2a}\rho_V '' + k_{1}\rho_V -2Gg_V '\right)	\\	\frac{1}{a} k_{2}\rho_V-\varepsilon_2 Gl_V \end{pmatrix}; \nonumber
			\\
			\begin{pmatrix} D_V T  \\	D_V W_1 \\	D_V N \\ D_V W_2 \end{pmatrix}
			&=  \begin{pmatrix}	- \alpha_V & \varphi_V & 0 &  \psi_V \\ -\beta_V & 0 & \varepsilon_1 \varphi_V & \frac{1}{a}\psi_V' \\
			0 & -\varepsilon_1\beta_V & \alpha_V & \frac{\varepsilon_1 }{a}\delta_V \\	\frac{\varepsilon_1 \varepsilon_2}{a}\delta_V & -\frac{\varepsilon_1 \varepsilon_2}{a}\psi_V' & \varepsilon_2 \psi_V & 0 \end{pmatrix} \cdot 
			\begin{pmatrix}	T  \\ W_1 \\  N \\ W_2 \end{pmatrix}; 
			\label{eq-derivation3}			
	\end{align}
where $\alpha_V, \beta_V$ and $\delta_V$ are given in \eqref{eq-notacion1}. We now restrict our definition of Lie bracket only on the set $\mathfrak{X}_{\mathcal{P}}^{*}(\gamma)$, which will be enough for our purposes.

\begin{proposition}\label{prop-lie-bracket}
	Let $\gamma$ be a null curve in $\Lambda$ and consider $[\cdot,\cdot]_{\gamma}:\mathfrak{X}^*_{\mathcal{P}}(\gamma)\times \mathfrak{X}^*_{\mathcal{P}}(\gamma)\rightarrow \mathfrak{X}^*_{\mathcal{P}}(\gamma)$ the map given by
	$$[\v_{1},\v_{2}]_{\gamma}=D_{\v_{1}}\v_{2}-D_{\v_{2}}\v_{1}.$$  Then:
	\begin{enumerate}[(a)]
		\item $[,]_\gamma$ is well defined, i.e., if $\v_{1},\v_{2}\in\mathfrak{X}^*_{\mathcal{P}}(\gamma)$ then $[\v_{1},\v_{2}]_\gamma\in \mathfrak{X}^*_{\mathcal{P}}(\gamma)$.
		\item $[\v_{1},\v_{2}]_{\gamma}(f)=\v_{1}\v_{2}(f)-\v_{2}\v_{1}(f)$ for all $f\in \mathcal{P}$.
		\item $[,]_{\gamma}$ is skew-symmetric.
		\item For $V_1,V_2\in \mathfrak{X}^*_{\mathcal{P}}(\gamma)$ and $\u\in \mathfrak{X}_{\mathcal{P}}(\gamma)$ we have 
		\begin{equation*}\label{formula-curvatura}
			D_{[V_1,V_2]}U-D_{V_1}D_{V_2}U+D_{V_2}D_{V_1}U=G \left(\left\langle U,V_{1} \right\rangle V_{2}-\left\langle U,V_{2} \right\rangle V_{1}\right).
		\end{equation*}
		\item $[,]_{\gamma}$ satisfies the Jacobi identity.
		\item $[,]_{\gamma}$ is closed for elements in $T_{\mathcal{P},\gamma}(\Lambda)$, i.e., if $\v,\u \in T_{\mathcal{P},\gamma}(\Lambda)$, then $[\v,\u]_{\gamma}\in T_{\mathcal{P},\gamma}(\Lambda)$.
	\end{enumerate}
\end{proposition}
\begin{proof} 
	Given two vector fields $\v_1=f_1 T + h_1 W_1 + g_1 N + l_1 W_2$ and $\v_2=f_2 T + h_2 W_1 + g_2 N + l_2 W_2$ in $\mathfrak{X}^*_{\mathcal{P}}(\gamma)$, set $\v_{12}=[\v_1,\v_2]_\gamma$. The components $g_{12}$ and  $h_{12}$ of $\v_{12}$ are given by
	\begin{equation*}\label{gh12}
		\begin{aligned}
			g_{12} &= \v_{1}(g_2)-\v_{2}(g_1)-\dfrac{1}{a}\left(\varphi_1 g_2' -\varphi_2 g_1'\right) + \left(\alpha_1 g_2 - \alpha_2 g_1\right) + \varepsilon_2 \left(l_2\psi_1 - l_1\psi_2\right), \\
			h_{12} &= \dfrac{\varepsilon_1}{a}\left(\v_{2}(g_1')-\v_{1}(g_2')\right)+\left(\varphi_1 f_2 -\varphi_2 f_1 \right) + \varepsilon_1\left(\beta_2 g_1 - \beta_1 g_2\right) + \dfrac{\varepsilon_1 \varepsilon_2}{a} \left(l_1\psi_2' - l_2\psi_1'\right).
		\end{aligned}
	\end{equation*}
	Since $\v_{1},\v_{2}\in\mathfrak{X}^*_{\mathcal{P}}(\gamma)$, they verify the formulas \eqref{eq-compatibilidad} and \eqref{formulas-compatibilidad} which, together with definitions of $\alpha_i$ and $\beta_i$, lead to the relation  $g'_{12}=-a\varepsilon_1 h_{12}$. The latter is the condition equivalent to $[\v_{1},\v_{2}]_{\gamma}\in\mathfrak{X}^*_{\mathcal{P}}(\gamma)$, thus proving \textit{(a)}. To prove \textit{(b)}, it is sufficient to check the same equality solely for the generators $k_{1}$ and $ k_{2}$ of the algebra $\mathcal{P}$. We shall calculate the expressions of $\varphi_{12}$, $\psi_{12}$ and $\rho_{12}$ (corresponding functions to the Lie bracket $[V_1,V_2]_{\gamma}$), by means of $\varphi_i , \psi_i$  and $\rho_i$ (corresponding functions to vector fields $V_i$).
	If $V=f_V T + h_V W_1 + g_V N + l_V W_2 \in\mathfrak{X}_{\mathcal{P}}(\gamma)$ is any vector field, we have:
			\begin{equation}\label{F}
				\begin{aligned}
				\varphi_V &=af_V + h_V ' -\varepsilon_1 k_{1}g_V, \\
				\psi_V &=l_V ' +\varepsilon_2 k_{2} g_V , \\
				\rho_V &= -af_V '+2a k_{1}h_V -a k_{2} l_V-\varphi_V ' +\varepsilon_1 k_{1}g_V', \\
				 h_V &=-\dfrac{\varepsilon_1}{a}g_V ' \; \text{if} \; V\in\mathfrak{X}^*_{\mathcal{P}}(\gamma).
				\end{aligned}
			\end{equation}
	Bearing in mind relations \eqref{F} and expressions of $V(k_{1})$ and $V( k_{2})$ obtained in Lemma \ref{lema-variation2}, we deduce:
			\begin{equation}\label{G}
				\begin{aligned}
				\varphi_{D_{V_1}V_2} &= af_{D_{V_1}V_2} + h_{D_{V_1}V_2}  ' -\varepsilon_1 k_{1}g_{D_{V_1}V_2} \\
						 & = V_1 (\varphi_2)-\dfrac{1}{2a}\rho_1\varphi_2 -\dfrac{1}{a}\varphi_1\rho_2 -\dfrac{1}{a}\varphi_1\varphi_2 ' -\dfrac{\varepsilon_1 \varepsilon_2}{a}\psi_1 '\psi_2 -\dfrac{\varepsilon_1}{a}Gg_1 'g_2; \\
				 \psi_{D_{V_1}V_2} &= l_{D_{V_1}V_2} ' +\varepsilon_2 k_{2} g_{D_{V_1}V_2} \\
						 & = V_1 (\psi_2) -\dfrac{1}{2a}\rho_1\psi_2 -\dfrac{1}{a}\rho_2\psi_1 +\dfrac{1}{a}\left(\psi_1 '\varphi_2 -\psi_1\varphi_2 ' \right) + Gl_1g_2; \\
				 \rho_{D_{V_1}V_2} &= -af_{D_{V_1}V_2} ' + 2a k_{1}h_{D_{V_1}V_2} -a k_{2} l_{D_{V_1}V_2} -\varphi _{D_{V_1}V_2} ' +\varepsilon_1 k_{1} g_{D_{V_1}V_2} ' \\
						 &= V_1 (\rho_2) -\dfrac{1}{a}\rho_1\rho_2 + \dfrac{1}{a}\left(\rho_1 '\varphi_2 + \varphi_1\rho_2 '\right) + \dfrac{1}{a}\left(\varphi_1 ''\varphi_2 + \varphi_2 ''\varphi_1\right)+2k_{1}\varphi_1\varphi_2 \\
						 & - k_{2}\left(\varphi_1\psi_2 + \varphi_2\psi_1\right)+\dfrac{\varepsilon_1\varepsilon_2}{a}\psi_1 ' \psi_2 ' + G\left(\dfrac{\varepsilon_1}{a}\left(g_1g_2 '' + g_2g_1 ''\right)+2\varepsilon_1 k_{1}g_1g_2 -a\varepsilon_2 l_1l_2\right).		 
				\end{aligned}
			\end{equation}
	Because of symmetry of formulas \eqref{G}, deleting terms with repeated factors and rearranging the other, we obtain
			\begin{equation}\label{H}
				\begin{aligned} 
				\varphi_{12} & = V_1 (\varphi_2)-V_2 (\varphi_1)+\dfrac{1}{2a}(\rho_1\varphi_2 - \rho_2\varphi_1)  \\
				&\quad +\dfrac{1}{a}(\varphi_1'\varphi_2 -\varphi_1\varphi_2')+\dfrac{\varepsilon_1 \varepsilon_2}{a}(\psi_1\psi_2 ' -\psi_1 '\psi_2) +\dfrac{\varepsilon_1}{a}G(g_2'g_1 - g_2 g_1'), \\
				 \psi_{12} 	 & = V_1 (\psi_2)-V_2(\psi_1)+\dfrac{1}{2a}(\rho_1\psi_2 -\rho_2\psi_1) \\
				 &\quad +\dfrac{1}{a}(\psi_1'\varphi_2 -\psi_1\varphi_2' +\psi_2\varphi_1' -\psi_2' \varphi_1) + G(l_1g_2 - l_2g_1),\\
				 \rho_{12}   &= V_1 (\rho_2) -V_2(\rho_1).	 
				\end{aligned} 
			\end{equation}
	From Lemma \ref{lema-variation2} we obtain
			\begin{equation}\label{J}
				\begin{aligned}
				\left[V_1 , V_2\right]_{\gamma}(k_{1}) & = \dfrac{1}{a^2}\varphi_{12}''' +\dfrac{1}{a}\left(k_{1}'\varphi_{12} +2k_{1}\varphi_{12}'\right)
								   -\dfrac{1}{a}\left( k_{2} '\psi_{12}+2 k_{2}\psi_{12}'\right) \\  
								   &\quad +\dfrac{1}{a}\left(\dfrac{1}{2a}\rho_{12}'' +k_{1}\rho_{12}-2Gg_{12}'\right). 
				\end{aligned}
			\end{equation}
	Expanding each earlier term by using $V(f)''' = V(f''')-\dfrac{3}{2a}\rho_V f''' -\dfrac{3}{2a}\rho_V ' f'' -\dfrac{1}{2a}\rho_V '' f'$ and other properties it follows that
	\begin{align}
			\dfrac{1}{a^2}\varphi_{12}''' & = V_1 \left(\dfrac{1}{a^2}\varphi_2 '''\right) - V_2 \left(\dfrac{1}{a^2}\varphi_1 '''\right) + \dfrac{1}{a^3}\left(\rho_2\varphi_1 ''' -\rho_1\varphi_2 '''\right) + \dfrac{1}{a^3}\left(\rho_1 ''\varphi_2 ' - \rho_2 ''\varphi_1 '\right) \nonumber \\
				&\quad +\dfrac{1}{2a^3}\left(\rho_1 '''\varphi_2 -\rho_2 '''\varphi_1\right) +\dfrac{1}{a^3}\left(\varphi_1^{(4)}\varphi_2 - \varphi_1\varphi_2^{(4)}+2\left(\varphi_1 '''\varphi_2 ' -\varphi_2 '''\varphi_1 '\right)\right) \label{K} \\
				&\quad+\dfrac{\varepsilon_1\varepsilon_2}{a^3}\left(\psi_1\psi_2 ^{(4)}-\psi_2\psi_1^{(4)}+2\left(\psi_1 '\psi_2 ''' - \psi_2 '\psi_1 '''\right)\right)+\dfrac{\varepsilon_1}{a^3}G(g_2'g_1 - g_2 g_1')'''; \nonumber 
			\end{align}
			\begin{align}	
				\dfrac{1}{a}\left(k_{1}'\varphi_{12}+2k_{1}\varphi_{12}'\right) & = V_1 \left(\dfrac{1}{a}\left(k_{1}'\varphi_2 +2 k_{1}\varphi_2 '\right)\right) - V_2 \left(\dfrac{1}{a}\left(k_{1}'\varphi_1 +2 k_{1}\varphi_1 '\right)\right) \nonumber  \\
				&\quad + \dfrac{1}{a}\left(V_2(2 k_{1})\varphi_1 ' +V_2 (k_{1}')\varphi_1 -V_1 (2 k_{1})\varphi_2 ' - V_1 (k_{1}')\varphi_2 \right) \nonumber \\
				&\quad  +\dfrac{k_{1}}{a^2}\left(\rho_1 '\varphi_2 -\rho_2 '\varphi_1\right)+\dfrac{2 k_{1}}{a^2}\left(\varphi_2\varphi_1 '' -\varphi_1\varphi_2 ''\right)+\dfrac{2\varepsilon_1\varepsilon_2 k_{1}}{a^2}\left(\psi_1\psi_2 '' -\psi_2\psi_1 ''\right) \label{L}  \\
				&\quad + \dfrac{k_{1}'}{2a^2}\left(\rho_1\varphi_2 -\rho_2\varphi_1\right) + \dfrac{k_{1}'}{a^2}\left(\varphi_2\varphi_1 ' -\varphi_1\varphi_2 '\right)+\dfrac{\varepsilon_1\varepsilon_2 k_{1}'}{a^2}\left(\psi_1\psi_2 ' -\psi_2\psi_1 '\right)  \nonumber \\
				&\quad+\dfrac{2\epsilon_1 k_{1}}{a^2}G(g_2'g_1 - g_2 g_1')' + \dfrac{\varepsilon_1 k_{1}'}{a^2}G(g_2'g_1 - g_2 g_1'); \nonumber 
			\end{align}
			\begin{align}
				-\dfrac{1}{a}\left( k_{2} '\psi_{12}+2 k_{2}\psi_{12}'\right) & = V_1 \left(-\dfrac{1}{a}\left( k_{2} '\psi_2 +2 k_{2}\psi_2 '\right)\right) - V_2 \left(-\dfrac{1}{a}\left( k_{2}'\psi_1 +2 k_{2}\psi_1 '\right)\right) \nonumber \\
				&\quad  - \dfrac{1}{a}\left(V_2(2 k_{2})\psi_1 ' +V_2 ( k_{2} ')\psi_1 -V_1 (2 k_{2})\psi_2 ' - V_1 ( k_{2}')\psi_2 \right) \nonumber \\
				&\quad -\dfrac{ k_{2}}{a^2}\left(\rho_1 '\psi_2 -\rho_2 '\psi_1\right)-\frac{k_{2}'}{2a^{2}}(\rho_{1}\psi_{2}-\rho_{2}\psi_{1})  \label{M} \\
				&\quad  -\dfrac{2 k_{2}}{a^2}\left(\varphi_2\psi_1 '' -\psi_1\varphi_2 '' +\psi_2\varphi_1 '' -\psi_2 ''\varphi_1\right) \nonumber  \\
				&\quad -\dfrac{ k_{2} '}{a^2}\left(\psi_1 '\varphi_2  -\psi_1\varphi_2 ' +\psi_2\varphi_1 ' -\psi_2 '\varphi_1\right)  \nonumber \\
				&\quad  - \dfrac{ k_{2} '}{a}G(l_1g_2 - l_2g_1) - \dfrac{2 k_{2}}{a}G(l_1g_2 - l_2g_1)'. \nonumber 
	\end{align}
	Based on the expressions $V(f)'' =V(f'')-\dfrac{1}{a}\rho_V f'' -\dfrac{1}{2a}\rho_V 'f'$ and $V(f)' =V(f')-\dfrac{1}{2a}\rho_V f'$ it is also easy to establish that
		\begin{align}
			\dfrac{1}{2a^2}\rho_{12}'' & = V_1 \left(\dfrac{1}{2a^2}\rho_2 ''\right) - V_2 \left(\dfrac{1}{2a^2}\rho_1 ''\right) +\dfrac{1}{2a^3}\left(\rho_2\rho_1 '' -\rho_1\rho_2 ''\right), \label{N} 
			\\
			\dfrac{k_{1}}{a}\rho_{12} & = V_1 \left(\dfrac{k_{1}}{a}\rho_2 \right) - V_2 \left(\dfrac{k_{1}}{a}\rho_1 \right) +V_2\left(\dfrac{k_{1}}{a}\right)\rho_1-V_1\left(\dfrac{k_{1}}{a}\right)\rho_2,  \label{O}
			\\
			-\dfrac{2G}{a}g_{12}' & = V_1 \left(-\dfrac{2G}{a}g_2 ' \right) - V_2 \left(-\dfrac{2G}{a}g_1 ' \right) +\dfrac{G}{a^2}\left(\rho_1 g_2' -\rho_2 g_1 '\right) \nonumber \\
					   & +\dfrac{2G}{a^2}\left(\varphi_1 'g_2 ' -\varphi_2 ' g_1 ' +\varphi_1 g_2 '' -\varphi_2 g_1 ''\right) -\dfrac{2G}{a}\left(\alpha_1 'g_2 + \alpha_1 g_2 ' - \alpha_2 'g_1 - \alpha_2 g_1 ' \right) \label{P}\\
					   & -\dfrac{2\varepsilon_2 G}{a}\left(\psi_1 'l_2 +\psi_1 l_2 ' -l_1 '\psi_2 - l_1\psi_2 '\right). \nonumber
		\end{align}
	When adding up \eqref{J}, \eqref{K}, \eqref{L}, \eqref{M}, \eqref{N}, \eqref{O} and \eqref{P}, and making a long but easy computation we obtain
	$$ [V_1 ,V_2]_{\gamma}(k_{1})=V_1 V_2 (k_{1}) -V_2 V_1 (k_{1}).$$
	Using again Lemma \ref{lema-variation2} we have
			\begin{equation}
				\begin{aligned}\label{Q}
				\left[V_1 , V_2\right]_{\gamma}( k_{2}) & = \dfrac{\varepsilon_1\varepsilon_2}{a^2}\psi_{12}''' +\dfrac{\varepsilon_1\varepsilon_2}{a}\left(k_{1}'\psi_{12} +2 k_{1}\psi_{12}'\right)
								   +\dfrac{1}{a}\left( k_{2} '\varphi_{12}+2 k_{2}\varphi_{12}'\right) \\
								   &\quad  +\dfrac{1}{a}\left( k_{2}\rho_{12}-a\varepsilon_2 Gl_{12}\right). 
				\end{aligned}
			\end{equation}
	In the same way as \eqref{H}, we can compute the terms of \eqref{Q},
			\begin{align}
				\dfrac{\varepsilon_1\varepsilon_2}{a^2}\psi_{12}''' & = V_1 \left(\dfrac{\varepsilon_1\varepsilon_2}{a^2}\psi_2 '''\right) - V_2 \left(\dfrac{\varepsilon_1\varepsilon_2}{a^2}\psi_1 '''\right) \nonumber \\
				&\quad  + \dfrac{\varepsilon_1\varepsilon_2}{a^3}\left(\rho_1 ''\psi_2 ' - \rho_2 ''\psi_1 '\right) +\dfrac{\varepsilon_1\varepsilon_2}{2a^3}\left(\rho_1 '''\psi_2 -\rho_2 '''\psi_1\right) \label{R} \\
				&\quad +\dfrac{\varepsilon_1\varepsilon_2}{a^3}\left[\psi_1^{(4)}\varphi_2 -\psi_2^{(4)}\varphi_1 +2\left(\psi_1 '''\varphi_2 ' -\psi_2 '''\varphi_1 '\right)+2\left(\psi_2 '\varphi_1 ''' -\psi_1 '\varphi_2 '''\right)\right. \nonumber \\
				&\quad\left. +\left(\psi_2\varphi_1^{(4)}-\psi_1\varphi_2^{(4)}\right)\right]+\dfrac{\varepsilon_1\varepsilon_2}{a^2}G(l_1g_2 - l_2g_1)'''; 
				\nonumber 
			\end{align}
			\begin{align}	
				\dfrac{\varepsilon_1\varepsilon_2}{a}\left(k_{1}'\psi_{12}+2 k_{1} \psi_{12}'\right) & = V_1 \left(\dfrac{\varepsilon_1\varepsilon_2}{a}\left(k_{1}'\psi_2 +2 k_{1}\psi_2 '\right)\right) - V_2 \left(\dfrac{\varepsilon_1\varepsilon_2}{a}\left(k_{1}'\psi_1 +2 k_{1}\psi_1 '\right)\right) \nonumber \\
				&\quad + \dfrac{\varepsilon_1\varepsilon_2}{a}\left(V_2(2 k_{1})\psi_1 ' +V_2 (k_{1}')\psi_1 -V_1 (2 k_{1})\psi_2 ' - V_1 (k_{1}')\psi_2 \right) \nonumber \\
				&\quad +\dfrac{\varepsilon_1\varepsilon_2 k_{1}}{a^2}\left(\rho_1 '\psi_2 -\rho_2 '\psi_1\right)+\dfrac{\varepsilon_1\varepsilon_2 k_{1}'}{2a^2}\left(\rho_1\psi_2 -\rho_2\psi_1\right)  \label{S} \\
				&\quad + \frac{2\varepsilon_1\varepsilon_2 k_{1}}{a^2}\left(\varphi_2\psi_1 '' -\varphi_1\psi_2 '' +\psi_2\varphi_1 '' - \psi_1\varphi_2''\right)\nonumber \\
				&\quad + \dfrac{\varepsilon_1\varepsilon_2 k_{1}'}{a^2}\left(\varphi_2\psi_1 ' -\psi_1\varphi_2' +\psi_2\varphi_1 ' -\psi_2 '\varphi_1\right) \nonumber \\
				&\quad +\frac{\varepsilon_1\varepsilon_2 G}{a}\left[2 k_{1} (l_1g_2 - l_2g_1)' +k_{1}'(l_1g_2 - l_2g_1)\right];
				\nonumber \\
				\dfrac{1}{a}\left( k_{2} '\varphi_{12}+2 k_{2}\varphi_{12}'\right) & = V_1 \left(\dfrac{1}{a}\left( k_{2} '\varphi_2 +2 k_{2}\varphi_2 '\right)\right) - V_2 \left(\dfrac{1}{a}\left( k_{2}'\varphi_1 +2 k_{2}\varphi_1 '\right)\right) \nonumber \\
				&\quad +\dfrac{1}{a}\left(V_2(2 k_{2})\varphi_1 ' +V_2 ( k_{2} ')\varphi_1 -V_1 (2 k_{2})\varphi_2 ' - V_1 ( k_{2}')\varphi_2 \right)  \nonumber \\
				&\quad +\dfrac{ k_{2}}{a^2}\left(\rho_1 '\varphi_2 -\rho_2 '\varphi_1\right) +\dfrac{ k_{2} '}{2a^2}\left(\rho_1\varphi_2 -\rho_2\varphi_1\right) \nonumber  \\ 
				&\quad +\dfrac{2 k_{2}}{a^2}\left(\varphi_2\varphi_1 '' -\varphi_1\varphi_2 ''\right) +\dfrac{2\varepsilon_1\varepsilon_2  k_{2}}{a^2}\left(\psi_1\psi_2 '' -\psi_1 ''\psi_2\right) \label{T} \\
				&\quad + \dfrac{ k_{2} '}{a^2}\left(\varphi_1 '\varphi_2 -\varphi_1\varphi_2 '\right) + \dfrac{\varepsilon_1\varepsilon_2  k_{2} '}{a^2}\left(\psi_1\psi_2 ' -\psi_1 '\psi_2\right) \nonumber \\
				&\quad +\dfrac{\varepsilon_1 G}{a^2}\left[2 k_{2} (g_2'g_1 - g_2 g_1')'+ k_{2}'(g_2'g_1 - g_2 g_1')\right]; 
				\nonumber 
			\end{align}
			\begin{align}	
				\dfrac{ k_{2}}{a}\rho_{12} & = V_1 \left(\dfrac{ k_{2}}{a}\rho_2 \right) - V_2 \left(\dfrac{ k_{2}}{a}\rho_1 \right) +V_2\left(\dfrac{ k_{2}}{a}\right)\rho_1-V_1\left(\dfrac{k_{1}}{a}\right)\rho_2; 
				\label{U} 
				\\
				-\varepsilon_2 Gl_{12} & = V_1 \left(-\varepsilon_2 Gl_2 \right) - V_2 \left(-\varepsilon_2 Gl_1 \right) -\varepsilon_2 G\left(f_2\psi_1 -f_1\psi_2\right) \nonumber \\
				&\quad +\dfrac{\varepsilon_1\varepsilon_2}{a^2}\left(g_2 '\psi_1 ' -g_1 '\psi_2 '\right)-\dfrac{\varepsilon_1\varepsilon_2 G}{a}\left(g_2\delta_1 -g_1\delta_2\right). \label{V} 
			\end{align}
	After some work, it also follows from \eqref{R}, \eqref{S}, \eqref{T}, \eqref{U}  and \eqref{V} that 
	$$ [V_1 ,V_2]_{\gamma}( k_{2})=V_1 V_2 ( k_{2}) -V_2 V_1 ( k_{2}).$$
	 The paragraph \textit{(c)} is a direct consequence of the definition, and  \textit{(f)} is also trivial taking into account the expression for $\rho_{12}$ given in  \eqref{H}.
	 The paragraph \textit{(d)} follows from a straightforward computation. Finally, to prove \textit{(e)}, let us denote $R(V_{1},V_{2})U=D_{[V_1,V_2]}U-D_{V_1}D_{V_2}U+D_{V_2}D_{V_1}U$. By using \textit{(d)} we obtain
	 $$[[V_1,V_2],V_3]+[[V_2,V_3],V_1]+[[V_3,V_1],V_2]=R(V_1,V_2)V_3+R(V_2,V_3)V_1+R(V_3,V_1)V_2=0.$$
\end{proof}
In particular, Proposition \ref{prop-lie-bracket} entails that the set of $\mathcal{P}$-local pseudo arc-length preserving variation vector fields  $T_{\mathcal{P},\gamma}(\Lambda)$ is a Lie subalgebra of the Lie algebra of the $\mathcal{P}$-local vector fields $(\mathfrak{X}^*_{\mathcal{P}}(\gamma),[,]_\gamma)$.

Before turning to study the geometric hierarchies of null curve flows in section \ref{hamiltonian-structure}, we point out that equations \eqref{eq-symplectic} and \eqref{eq-cosymplectic} are particularly noteworthy when vector fields locally preserve the pseudo arc-length parameter and the curvature $G$ vanishes. Equations \eqref{eq-symplectic} and \eqref{eq-cosymplectic} may be rewritten as
\begin{equation}\label{eq-operators}
\begin{aligned}
	(\varphi_V,\psi_V)&=A(h_V,l_V)=\left(A_1(h_V,l_V),A_2(h_V,l_V)\right), \\
	(V(k_{1}),V( k_{2}))&=B(\varphi_V,\psi_V)=\left(B_1(\varphi_V,\psi_V),B_2(\varphi_V,\psi_V)\right),
\end{aligned}
\end{equation}	
where 
\begin{equation}
A=\begin{pmatrix}A_1 \\ A_2 \end{pmatrix}=\begin{pmatrix}  \frac{a}{2}\omega(k_{1})  & -\frac{a}{2}D_{\sigma}^{-1} k_{2} \\ -\varepsilon_1 \varepsilon_2 a k_{2} D_{\sigma}^{-1}  & D_{\sigma} \\ \end{pmatrix}; \qquad
B=\begin{pmatrix}B_1 \\ B_2 \end{pmatrix}=\frac{1}{a}\begin{pmatrix}\theta(k_{1}) & -S( k_{2}) \\	S( k_{2}) & \varepsilon_1 \varepsilon_2 \theta(k_{1}) \end{pmatrix},
\end{equation}
with $\omega(k_{1})=\frac{1}{a}D_{\sigma}+k_{1}D_{\sigma}^{-1}+D_{\sigma}^{-1} k_{1}$, $\theta(k_{1})=\frac{1}{a}D_\sigma^3+k_{1}D_\sigma+D_{\sigma}k_{1}$ and $S( k_{2})=D_{\sigma} k_{2} +  k_{2} D_{\sigma}$.
It should be remarked that $A$ and $B$ come very close to being the symplectic and cosymplectic operators, respectively, for (up to scaling) the Hirota-Satsuma system (see \cite{fuchssteiner_lie_1982,oevel_integrability_1983}). Equation \eqref{eq-operators} can also be regarded as
\begin{equation}\label{split1}
\begin{aligned}
	(\varphi_V,-\psi_V)&=J(2h_V,-\varepsilon_{1}\varepsilon_{2}l_V)=\left(J_1(2h_V,-\varepsilon_{1}\varepsilon_{2}l_V),J_2(2h_V,-\varepsilon_{1}\varepsilon_{2}l_V)\right), \\
	(V(k_{1}),V( k_{2}))&=\Theta(\varphi_V,-\psi_V)=\left(\Theta_1(\varphi_V,-\psi_V),\Theta_2(\varphi_V,-\psi_V)\right),
\end{aligned}
\end{equation}
where 
\begin{equation}\label{symplectic-cosymplectic}
J=\begin{pmatrix}J_1 \\ J_2 \end{pmatrix}=-\frac{\varepsilon_{1}\varepsilon_{2}}{2}\begin{pmatrix}  -\frac{\varepsilon_{1}\varepsilon_{2} a}{2}\omega(k_{1})  & -a D_{\sigma}^{-1} k_{2} \\ -a k_{2} D_{\sigma}^{-1}  & -2 D_{\sigma} \\ \end{pmatrix}; \qquad
\Theta=\begin{pmatrix}\Theta_1 \\ \Theta_2 \end{pmatrix}=\frac{1}{a}\begin{pmatrix}\theta(k_{1}) & S( k_{2}) \\	S( k_{2}) & -\varepsilon_1 \varepsilon_2 \theta(k_{1}) \end{pmatrix},
\end{equation}

It is now therefore evident that $J$ and $\Theta$ are the symplectic and cosymplectic operators respectively for a rescaling of the HS-cKdV system. They have been obtained in a natural way using projections onto the screen bundle of both, the variation vector field $V$ and its covariant derivative $\nabla_T V$. This allows automatically to determine the recursion operator  $R=\Theta\circ J$, and the crucial relation
\begin{equation}\label{eq-recursion}
 (V(k_{1}),V( k_{2}))=R(2h_{V},-\varepsilon_{1}\varepsilon_{2}l_{V}).
\end{equation}

Somewhat analogous relationships were obtained between curve evolution in $3$-dimensional Riemannian manifolds in \cite{	  beffa_integrable_2002} (or more generally in $n$-dimensional Riemannian manifold with constant curvature in \cite{sanders_integrable_2003}) and the mKdV system. The above connection together with the availability of the Lie bracket provided by Proposition \ref{prop-lie-bracket} will be employed to study the integrability of null curve evolution in the next section. 

\section{Geometric hierarchies of null curve flows}\label{hamiltonian-structure}
The background given in \cite{del_amor_hamiltonian_2014} for the $3$-dimensional case used to construct a commuting hierarchy for null curve evolutions can also be well adapted to the $4$-dimensional case. Consider $\Lambda$ the space of pseudo arc-length parametrized null curves in the pseudo-Euclidean space $\mathbb{R}^{4}_{q}$. A map $\f:\Lambda \rightarrow \mathcal{C}^\infty(I,\mathbb{R})$ is referred to a scalar field on $\Lambda$ and $\f(\gamma)$ will be also denoted by $\f_\gamma$.
Let $\mathcal{A}$ be the algebra of $\mathcal{P}$-valued scalar fields on $\Lambda$, i.e., if $\f\in \mathcal{A}$, then $\f_\gamma\in \mathcal{P}$ for all $\gamma\in \Lambda$. In this sense, we will also understand the curvatures scalar fields $\k_{\bm{1}},\k_{\bm{2}}:\Lambda \rightarrow \mathcal{C}^\infty(I,\mathbb{R})$  with its obvious meaning.

Similarly, a map $\V:\Lambda \rightarrow \cup_{\gamma\in \Lambda} T_\gamma\Lambda$ is referred to as a vector field on $\Lambda$, and $\V(\gamma)$ will be also denoted by $\V_\gamma$. We shall denote the set of tangent vector fields on $\Lambda$ as $\mathfrak{X}(\Lambda)$, and within we consider the subset $\mathfrak{X}_{\mathcal{A}}(\Lambda)$ of vector fields $\V$ such that $\V_{\gamma}\in \mathfrak{X}_{\mathcal{P}}(\gamma)$, namely, if we denote $\V=\f\T+\h\W_{\bm{1}}+\g \N+\l \W_{\bm{2}}$, then
\begin{equation}\label{vector-fields}
	\mathfrak{X}_{\mathcal{A}}(\Lambda)=\left\{\V\in \mathfrak{X}(\Lambda): \f,\h,\g,\l\in \mathcal{A};\, \h=-\frac{\varepsilon_1}{a}\g';\, \f=\frac{1}{2a }
	\left[\frac{\varepsilon_1}{a}\g''+\D^{-1}_{\sigma}(\varepsilon_1\k_{\bm{1}}'\g-a\k_{\bm{2}}\l)\right]\right\},
\end{equation}
where the derivative and anti-derivative operators act on scalar fields as $\f'(\gamma)=\f'_\gamma$ and $\D^{-1}_{\sigma}(\f)(\gamma)=D^{-1}_{\sigma}(\f_\gamma)$ respectively. Thus $\mathfrak{X}_{\mathcal{A}}(\Lambda)$ stands for the set of $\mathcal{A}$-local vector fields locally preserving the pseudo arc-length parameter and the causal character. These vector fields commute with the tangent vector field $\T$, so they will be called evolution vector fields. We also denote by $\bar{\mathfrak{X}}_{\mathcal{A}}(\Lambda)$ and $\mathfrak{X}^*_{\mathcal{A}}(\Lambda)$ the sets of vector fields $\V$ such that $\V_\gamma\in \mathfrak{X}_{\mathcal{P}}(\gamma)$ and $\V_\gamma\in \mathfrak{X}^*_{\mathcal{P}}(\gamma)$, respectively. Hence,
\begin{equation*}
\begin{aligned}
	\bar{\mathfrak{X}}_{\mathcal{A}}(\Lambda)&=\left\{\V=\f\T+\h\W_{\bm{1}}+\g \N+\l \W_{\bm{2}}: \f,\h,\g,\l\in \mathcal{A}\right\}. \\
	\mathfrak{X}^*_{\mathcal{A}}(\Lambda)&=\left\{\V\in \bar{\mathfrak{X}}_{\mathcal{A}}(\Lambda): \h=-\frac{\varepsilon_1}{a}\g'\right\}.
\end{aligned}
\end{equation*}
\begin{remark}
	In what follows, we shall operate with scalar fields and vector fields in the natural way, understanding that the result of the operation is again a scalar field or vector field. For instance, if $\V,\U$ are vector fields on $\Lambda$, then $\left\langle \V,\U \right\rangle$ is a scalar field, where $\left\langle \V,\U \right\rangle(\gamma)=\left\langle \V_\gamma,\U_\gamma \right\rangle$; or $\nabla_{\T} \V$ is again a vector field, where $\nabla_{\T} \V(\gamma)=\nabla_{\T_\gamma}\V_\gamma$ and so on.
\end{remark}
Hence, for $\V\in \mathfrak{X}^*_{\mathcal{A}}(\Lambda)$, the operator 
$D_{\V}:\bar{\mathfrak{X}}_{\mathcal{A}}(\Lambda) \rightarrow \bar{\mathfrak{X}}_{\mathcal{A}}(\Lambda)$ 
is defined as $(D_{\V}\U)(\gamma)=D_{\V_\gamma}\U_\gamma$. The operator $D_{\V}$ can also be described in other words when $V\in \mathfrak{X}_{\mathcal{A}}(\Lambda)$. Consider $\gamma$ be a null curve in $\Lambda$,  and suppose that $\V_{\gamma}(\sigma)=\frac{\partial \gamma}{\partial t}(\sigma,0)$, then  
$$\ds (D_{\V}\U)_\gamma(\sigma)=\left.\frac{D}{\partial t}\right|_{t = 0}\U_{\gamma_t}(\sigma).$$
In fact, the tensor derivation $D_{\V}$ is an extension of the Fréchet derivative defined in \eqref{frechet-derivative} for derivations to vector fields on the null curves space. In this way, this operator can be easily translated to the context of any other type of curves.

The Lie algebra structure on local vector fields locally preserving the causal character provided by Proposition \ref{prop-lie-bracket} (along a particular curve) can also be easily extended on the set $\mathfrak{X}^*_{\mathcal{A}}(\Lambda)$.

\begin{proposition}\label{prop-LieBracket2}
	The map  $[\cdot,\cdot]:\mathfrak{X}^*_{\mathcal{A}}(\Lambda)\times \mathfrak{X}^*_{\mathcal{A}}(\Lambda)\rightarrow \mathfrak{X}^*_{\mathcal{A}}(\Lambda)$  given by
	$$[\V,\U](\gamma)=[\V_{\gamma},\U_{\gamma}]_\gamma$$  is a Lie bracket verifying the following
	\begin{enumerate}[(a)]
		\item $[\V,\U](\f)=\V\U(\f)-\U\V(\f)$ for all $\f\in \mathcal{A}$.
		\item $[,]$ is closed for elements in $\mathfrak{X}_{\mathcal{A}}(\Lambda)$, i.e., if $\V,\U \in \mathfrak{X}_{\mathcal{A}}(\Lambda)$, then $[\V,\U]\in \mathfrak{X}_{\mathcal{A}}(\Lambda)$.
	\end{enumerate}
	Hence, $[,]$ is a Lie bracket, ($\mathfrak{X}^*_{\mathcal{A}}(\Lambda),[,])$ is a Lie algebra and the space of evolution vector fields $(\mathfrak{X}_{\mathcal{A}}(\Lambda),[,])$ is a Lie subalgebra of $\mathfrak{X}^*_{\mathcal{A}}(\Lambda)$.
\end{proposition}

Let us define $\mathop{\rm der}\nolimits (\mathcal{A})$ as the set of derivations on $\mathcal{A}$ defined in the natural way and  $\mathop{\rm der}\nolimits^*(\mathcal{A})$ the Lie subalgebra of all evolution derivations. In this setting, the elements of $\mathop{\rm der}\nolimits^*(\mathcal{A})$ are given by $\Partial_{(\p,\q)}$, with $\p,\q\in \mathcal{A}$, such that they are defined as usual by $\Partial_{(\p,\q)}\f(\gamma)=\Partial_{(\p_\gamma,\q_\gamma)}\f_\gamma$, for all $\f\in \mathcal{A}$ and $\gamma\in\Lambda$. Each vector field $\V$ on $\Lambda$ can be regarded as a derivation on $\mathcal{A}$, 
acting on the generators $\k_{\bm{1}}$ and $\k_{\bm{2}}$ in the following way: $$\V(\k_{\bm{1}})(\gamma)=\V_{\gamma}(\k_{\bm{1}\gamma}), \quad \V(\k_{\bm{2}})(\gamma)=\V_{\gamma}(\k_{\bm{2}\gamma}).$$

\begin{theorem}\label{main-theorem}
	The map $\Phi:\mathfrak{X}_{\mathcal{A}}(\Lambda) \rightarrow \mathop{\rm der}\nolimits^*(\mathcal{A})$ defined by $$\Phi(\V)=\Partial_{\left(\V(\k_{\bm{1}}),\V(\k_{\bm{2}})\right)}=\Partial_{\Theta(\bm{\varphi}_{\V},-\bm{\psi}_{\V})},$$ 
	where $\bm{\varphi}_{\V}=\varepsilon_{1}\left\langle \nabla_{\T} \V, \W_{\bm{1}} \right\rangle$, $\bm{\psi}_{\V}=\varepsilon_{2}\left\langle \nabla_{\T} \V, \W_{\bm{2}} \right\rangle$ and $\Theta$ is defined in \eqref{symplectic-cosymplectic}, is a one-to-one homomorphism of Lie algebras. In particular, $\V_{\bm{1}}$ and $\V_{\bm{2}}$ are commuting vector fields with respect to the Lie bracket defined by Proposition \ref{prop-LieBracket2} if and only if their corresponding curvature flows $\left(\V_{\bm{1}}(\k_{\bm{1}}),\V_{\bm{1}}(\k_{\bm{2}})\right)$ and $\left(\V_{\bm{2}}(\k_{\bm{1}}),\V_{\bm{2}}(\k_{\bm{2}})\right)$ commute with respect to the usual Lie bracket for scalar fields.
\end{theorem}
\begin{proof}
	 Since the map $\Phi$ is clearly linear, it is enough to prove that $\Phi$ keeps the Lie bracket, i.e., $\Phi([\V_{\bm{1}},\V_{\bm{2}}])=[\Phi(\V_{\bm{1}}),\Phi(\V_{\bm{2}})]$, the latter being equivalent to show that 
	 \begin{equation}\label{condicion-main-theorem}
		 \Partial_{\left([\V_{\bm{1}},\V_{\bm{2}}](\k_{\bm{1}}),[\V_{\bm{1}},\V_{\bm{2}}](\k_{\bm{2}})\right)}=\left[\Partial_{\left(\V_{\bm{1}}(\k_{\bm{1}}),\V_{\bm{1}}(\k_{\bm{2}})\right)},\Partial_{\left(\V_{\bm{2}}(\k_{\bm{1}}),\V_{\bm{2}}(\k_{\bm{2}})\right)}\right].
	 \end{equation} 
	 From the last equality of Proposition \ref{prop-lie-bracket}(b) we deduce
	 \begin{equation*}
		 \begin{aligned}
			\left[\V_{\bm{1}\gamma},\V_{\bm{2}\gamma}\right]_\gamma(\k_{\bm{1}\gamma})&=\V_{\bm{1}\gamma}\V_{\bm{2}\gamma}(\k_{\bm{1}\gamma})-\V_{\bm{2}\gamma}\V_{\bm{1}\gamma}(\k_{\bm{1}\gamma}) \\
			&=\Partial_{\left(\V_{\bm{1}\gamma}(\k_{\bm{1}\gamma}),\V_{\bm{1}\gamma}(\k_{\bm{2}\gamma})\right)}\V_{\bm{2}\gamma}(\k_{\bm{1}\gamma})-\Partial_{\left(\V_{\bm{2}\gamma}(\k_{\bm{1}\gamma}),\V_{\bm{2}\gamma}(\k_{\bm{2}\gamma})\right)}\V_{\bm{1}\gamma}(\k_{\bm{1}\gamma}),\\  
			\left[\V_{\bm{1}\gamma},\V_{\bm{2}\gamma}\right]_\gamma(\k_{\bm{2}\gamma})&=\V_{\bm{1}\gamma}\V_{\bm{2}\gamma}(\k_{\bm{2}\gamma})-\V_{\bm{2}\gamma}\V_{\bm{1}\gamma}(\k_{\bm{2}\gamma}) \\
			&=\Partial_{\left(\V_{\bm{1}\gamma}(\k_{\bm{1}\gamma}),\V_{\bm{1}\gamma}(\k_{\bm{2}\gamma})\right)}\V_{\bm{2}\gamma}(\k_{\bm{2}\gamma})-\Partial_{\left(\V_{\bm{2}\gamma}(\k_{\bm{1}\gamma}),\V_{\bm{2}\gamma}(\k_{\bm{2}\gamma})\right)}\V_{\bm{1}\gamma}(\k_{\bm{2}\gamma}),
		 \end{aligned}
	 \end{equation*}
	 and it is therefore satisfied
	 \begin{equation*}
		\begin{aligned}
			\left([\V_{\bm{1}\gamma},\V_{\bm{2}\gamma}]_\gamma(\k_{\bm{1}\gamma}),[\V_{\bm{1}\gamma},\V_{\bm{2}\gamma}]_\gamma(\k_{\bm{2}\gamma})\right)
			&=\Partial_{\left(\V_{\bm{1}\gamma}(\k_{\bm{1}\gamma}),\V_{\bm{1}\gamma}(\k_{\bm{2}\gamma})\right)}\left(\V_{\bm{2}\gamma}(\k_{\bm{1}\gamma}),\V_{\bm{2}\gamma}(\k_{\bm{2}\gamma})\right) \\
			&\quad -\Partial_{\left(\V_{\bm{2}\gamma}(\k_{\bm{1}\gamma}),\V_{\bm{2}\gamma}(\k_{\bm{2}\gamma})\right)}\left(\V_{\bm{1}\gamma}(\k_{\bm{1}\gamma}),\V_{\bm{1}\gamma}(\k_{\bm{2}\gamma})\right) \\
			&=\left[\left(\V_{\bm{1}\gamma}(\k_{\bm{1}\gamma}),\V_{\bm{1}\gamma}(\k_{\bm{2}\gamma})\right),\left(\V_{\bm{2}\gamma}(\k_{\bm{1}\gamma}),\V_{\bm{2}\gamma}(\k_{\bm{2}\gamma})\right)\right]
		\end{aligned}
	 \end{equation*}
	 for all $\gamma\in\Lambda$. Finally, formula \eqref{condicion-main-theorem} is followed from:
	 \begin{equation*}
		\left[\Partial_{\left(\V_{\bm{1}\gamma}(\k_{\bm{1}\gamma}),\V_{\bm{1}\gamma}(\k_{\bm{2}\gamma})\right)},\Partial_{\left(\V_{\bm{2}\gamma}(\k_{\bm{1}\gamma}),\V_{\bm{2}\gamma}(\k_{\bm{2}\gamma})\right)}\right]=\Partial_{\left[\left(\V_{\bm{1}\gamma}(\k_{\bm{1}\gamma}),\V_{\bm{1}\gamma}(\k_{\bm{2}\gamma})\right),\left(\V_{\bm{2}\gamma}(\k_{\bm{1}\gamma}),\V_{\bm{2}\gamma}(\k_{\bm{2}\gamma})\right)\right]}. 
	 \end{equation*}
	 
	In order to prove the injectivity we shall prove that $\Partial_{(\V(\k_{\bm{1}}),\V(\k_{\bm{2}}))}=0$ implies $\V=0$, which is equivalent to proving that $\V(\k_{\bm{1}})=\V(\k_{\bm{2}})=0$ implies $\V=0$. As a first step we will prove that if $\V(\k_{\bm{1}})=\V(\k_{\bm{2}})=0$ then  $\bm{\varphi}_{\V}=\bm{\psi}_{\V}=0$ and so,  from formula \eqref{formula-DV},  $\nabla_{\T}\V=0$. According to \eqref{split1} we have that	$$(\V(\k_{\bm{1}}),\V(\k_{\bm{2}}))=\Theta(\bm{\varphi}_{\V},-\bm{\psi}_{\V})=\frac{1}{a}\left(\theta(\k_{\bm{1}})\bm{\varphi}_{\V}-S(\k_{\bm{2}})\bm{\psi}_{\V},S(\k_{\bm{2}})\bm{\varphi}_{\V}+\varepsilon_{1}\varepsilon_{2}\theta(\k_{\bm{1}})\bm{\psi}_{\V}\right).$$
	For a scalar field $\f\in \mathcal{A}$ we denote by $\rk(\f)$ the order of the highest derivative (with respect to both $\k_{\bm{1}}$ or $\k_{\bm{2}}$) appearing in $\f$, i.e.,
	$$\rk(\f)=\mathop{\rm max}\nolimits \left\{i:\frac{\Partial\f}{\Partial\k_{\bm{1}}^{(i)}}\neq 0\quad\text{or}\quad\frac{\Partial\f}{\Partial\k_{\bm{2}}^{(i)}}\neq 0\right\}.$$ 
	Suppose that $\rk(\bm{\varphi}_{\V})=n\neq 0$, then we have that $\rk(\theta(\k_{\bm{1}})\bm{\varphi}_{\V})=n+3$. Since $\V(\k_{\bm{1}})=0$ it is necessarily obtained that $\rk(S(\k_{\bm{2}})\bm{\psi}_{\V})=n+3$, whence $\rk(\bm{\psi}_{\V})=n+2$. Accordingly, $\rk(\theta(\k_{\bm{1}})\bm{\psi}_{\V})=n+5$, which together with the equation $\V(\k_{\bm{2}})=0$ would lead to $\rk(\bm{\varphi}_{\V})=n+4$ and so a contradiction. Therefore the scalar field $\bm{\varphi}_{\V}$ is constant, $\bm{\varphi}_{\V}=c$, and it would verify the equation
	$$0=\frac{1}{a}\left(\theta(\k_{\bm{1}})\bm{\varphi}_{\V}-S(\k_{\bm{2}})\bm{\psi}\right)=\frac{1}{a}\left(c\k_{\bm{1}}'-S(\k_{\bm{2}})\bm{\psi}_{\V}\right).$$
	The latter equation is satisfied if and only if $c=0$ and $\bm{\psi}_{\V}=0$, i.e., $\nabla_{\T}\V=0$. Using the formula \eqref{eq1-derivadas}, the equation $\nabla_{\T}\V=0$ can be developed as
	\begin{equation}\label{eqs-constant-field}
		\begin{aligned}
			&\f_{\V} '-\k_{\bm{1}} \h_{\V} +\k_{\bm{2}} \l_{\V}=0, \\
			&a\f_{\V} +\h_{\V} ' -\varepsilon_1 \k_{\bm{1}} \g_{\V}=0, \\ 
			&\varepsilon_1 a\h_{\V} +\g_{\V} '=0, \\ 
			&\l_{\V} ' +\varepsilon_2\k_{\bm{2}} \g_{\V}=0.  \\
		\end{aligned}
	\end{equation} 
	Suppose that $\rk(\l_{\V})=n\neq 0$, then the equations \eqref{eqs-constant-field} give rise to the following implications
	$$\rk(\g_{\V})=n+1 \Rightarrow \rk(\h_{\V})=n+2 \Rightarrow \rk(\f_{\V})=n+3.$$
	Nevertheless, those orders of derivation represent a direct contradiction to the first equation in \eqref{eqs-constant-field} unless $\l=\g=\h=\f=0$.
\end{proof}

\begin{remark}\label{main-remark}
	From Theorem \ref{main-theorem} we have that $\mathop{\rm Im}\nolimits(\Phi)$ is a Lie subalgebra of the algebra   $\mathop{\rm der}\nolimits^*(\mathcal{A})$ of all evolution derivations. Thus, we conclude  that the algebra of evolution vector fields $\mathfrak{X}_{\mathcal{A}}(\Lambda)$ on $\Lambda$ can be regarded as a Lie subalgebra of the evolution derivations. 
\end{remark}
Consider the vector fields $\V_{\bm{0}}=b \T$ and $\V_{\bm{1}}=-ac\k_{\bm{1}} \T-2\varepsilon_{1}a^{2}c\N$ borrowed from Example \ref{example-constants}. Their flows $\gamma_{t}\in\Lambda$ are governed by the equations
\begin{align}
	\frac{d}{dt}(\gamma_t)&=\V_{\bm{0}\gamma_{t}}=b\T_{\gamma_{t}}, \\
	\label{NLIE}
	\frac{d}{dt}(\gamma_t)&=\V_{\bm{1}\gamma_{t}}=-ac \k_{\bm{1}\gamma_{t}}\T_{\gamma_{t}}-2 \varepsilon_1 a^2 c \N_{\gamma_{t}},
\end{align}
which in turn induce evolutions for the curvature functions $\k_{\bm{1}}$ and $\k_{\bm{2}}$ given by 
\begin{align}
	\label{curve-flow-zero}
	&\begin{cases}
		\frac{d}{dt}(\k_{\bm{1}\gamma_{t}})=\V_{\bm{0}\gamma_{t}}(\k_{\bm{1}\gamma_{t}})=b\k'_{\bm{1}\gamma_{t}} \\ 
		\frac{d}{dt}(\k_{\bm{2}\gamma_{t}})=\V_{\bm{0}\gamma_{t}}(\k_{\bm{2}\gamma_{t}})=b\k'_{\bm{2}\gamma_{t}} 
	\end{cases}\\
	\label{curve-flow-one}
	&\begin{cases}
		\frac{d}{dt}(\k_{\bm{1}\gamma_{t}})=\V_{\bm{1}\gamma_{t}}(\k_{\bm{1}\gamma_{t}})= c \left(\k_{\bm{1}\gamma_{t}}^{(3)}+3 a \k_{\bm{1}\gamma_{t}} \k_{\bm{1}\gamma_{t}}'+6 \epsilon_1 \epsilon _2 a \k_{\bm{2}\gamma_{t}}  \k'_{\bm{2}\gamma_{t}}\right)\\
		\frac{d}{dt}(\k_{\bm{2}\gamma_{t}})=\V_{\bm{1}\gamma_{t}}(\k_{\bm{2}\gamma_{t}})=-c \left(2 \k_{\bm{2}\gamma_{t}}^{(3)}+3 a \k_{\bm{1}\gamma_{t}} \k_{\bm{2}\gamma_{t}}'\right)
	\end{cases}
\end{align} 
Observe that the evolution equation \eqref{curve-flow-one} is a generalization for the Hirota-Satsuma equation \eqref{eq-HS} (through suitable constants). Besides, the flows associated to the curvatures given by $\V_{\bm{0}}$ and $\V_{\bm{1}}$ are basically the flows $\sigma_{0}$ and $\sigma_{1}$ given in \eqref{eq-flows-HS}.  
We refer to the equation \eqref{NLIE} induced by $\V_{\bm{1}}$ (that also appears in \cite{li_motions_2013}) as the \emph{null localized induction equation (NLIE)}. 
Theorem \ref{main-theorem} will be used below to obtain a recursion operator for NLIE,  and thereby prove its integrability. 

\begin{proposition}\label{prop-recursion}
	The operator $\R$ acting on symmetries $\V$ as follows 
	\begin{equation}
		\R(\V)=\mathcal{X}\left(\frac{1}{2}\V(\k_{\bm{1}}),-\varepsilon_{1}\varepsilon_{2}\V(\k_{\bm{2}})\right),
	\end{equation} 
	is a recursion operator for NLIE.	
\end{proposition}
\begin{proof}
	Let $\U=\R(\V)$. Then by definition and making use of the equation \eqref{eq-recursion} we obtain 	$$(\U(\k_{\bm{1}}),\U(\k_{\bm{2}}))=R(2\h_{\U},-\varepsilon_{1}\varepsilon_{2}\l_{\U})=R(\V(\k_{\bm{1}}),\V(\k_{\bm{2}})).$${}
	The result can be easily deduced as a consequence of Theorem \ref{main-theorem}.
\end{proof}

We now proceed with the construction of an infinite hierarchy of symmetries following the same scheme as in equation \eqref{eq-flows-HS}, 
\begin{equation}
	\V_{\bm{2n}}=\R^{n}\V_{\bm{0}};\quad \V_{\bm{2n+1}}=\R^{n}\V_{\bm{1}}.
\end{equation}
Then, we have
\begin{equation}
\begin{aligned}
	\V_{\bm{2}}=\R\V_{\bm{0}}&=\left(-\frac{b}{4a}\k_{\bm{1}}''-\frac{b}{8}\k_{\bm{1}}^{2}+\frac{\epsilon_{1}\epsilon_{2}b}{4}\k_{\bm{2}}^{2}+\frac{\epsilon_{1}c_{1}}{2a}\k_{\bm{1}}+c_{2}\right)\T+\frac{b}{2}\k_{\bm{1}}'\W_{\bm{1}} \\
	&\quad +\left(c_{1} -\frac{\epsilon_{1}a b}{2}\k_{\bm{1}}\right)\N-\epsilon_{1}\epsilon_{2}b\k_{\bm{2}}'\W_{\bm{2}},
\end{aligned}	
\end{equation}
and the corresponding curvature flow is
\begin{equation}
	\begin{aligned}
		\V_{\bm{2}}(\k_{\bm{1}})&=\frac{1}{8a^{2}}\left[2b\k_{\bm{1}}^{(5)}+\left(10ab\k_{\bm{1}}-4\epsilon_{1}c_{1}\right)\k_{\bm{1}}^{(3)}+20\epsilon_{1}\epsilon_{2}ab\k_{\bm{2}}\k_{\bm{2}}^{(3)}+20ab\k_{\bm{1}}'\k_{\bm{1}}''\right.\\
		&\quad \left. +20\epsilon_{1}\epsilon_{2}a b\k_{\bm{2}}'\k_{\bm{2}}''+\left(15a^{2}b\k_{\bm{1}}^{2}+10\epsilon_{1}\epsilon_{2}a^{2}b\k_{\bm{2}}^{2}-12\epsilon_{1}ac_{1}\k_{\bm{1}}+8a^{2}c_{2}\right)\k_{\bm{1}}' \right. \\
		&\quad \left. +\left(20\epsilon_{1}\epsilon_{2}a^{2}b\k_{\bm{1}}-24\epsilon_{2}ac_{1}\right)\k_{\bm{2}}\k_{\bm{2}}')\right] \\
		\V_{\bm{2}}(\k_{\bm{2}})&=\frac{1}{8a^{2}}\left[-8 b\k_{\bm{2}}^{(5)}+\left(8\epsilon_{1}c_{1}-20ab\k_{\bm{1}}\right)\k_{\bm{2}}^{(3)}-10ab\k_{\bm{1}}''\k_{\bm{2}}'-20ab\k_{\bm{1}}'\k_{\bm{2}}''\right.\\
		&\quad \left.+\left(10\epsilon_{1}\epsilon_{2}a^{2}b\k_{\bm{2}}^{2}-5a^{2}b\k_{\bm{1}}^{2}+12\epsilon_{1}ac_{1}\k_{\bm{1}}+8a^{2}c_{2}\right)\k_{\bm{2}}'\right].
	\end{aligned}
\end{equation}
Likewise, the next vector field in the hierarchy becomes
\begin{equation}
	\begin{aligned}
		\V_{\bm{3}}=\R\V_{\bm{1}}&=\left(-\frac{c}{4a}\k_{\bm{1}}^{(4)}-\frac{3c}{4}\k_{\bm{1}}\k_{\bm{1}}''-\frac{7c}{8}\left(\k_{\bm{1}}'\right)^{2}-\frac{5\epsilon_{1}\epsilon_{2}c}{2}\k_{\bm{2}}\k_{\bm{2}}''-\epsilon_{1}\epsilon_{2}c\left(\k_{\bm{2}}'\right)^{2}\right. \\
		&\qquad \left. -\frac{ac}{8}\k_{\bm{1}}^{3}+\frac{\epsilon_{1}c_{3}}{2a}\k_{\bm{1}}-\frac{3\epsilon_{1}\epsilon_{2}ac}{4}\k_{\bm{1}}\k_{\bm{2}}^{2}+c_{4}\right)\T \\
		&\quad+\left(\frac{c}{2}\k_{\bm{1}}^{(3)}+\frac{3ac}{2}\k_{\bm{1}}\k_{\bm{1}}'+3\epsilon_{1}\epsilon_{2}ac\k_{\bm{2}}\k_{\bm{2}}'\right)\W_{\bm{1}} \\
		&\quad +\left(-\frac{\epsilon_{1}ac}{2}\k_{\bm{1}}''-\frac{3\epsilon_{1}a^{2}c}{4}\k_{\bm{1}}^{2}-\frac{3\epsilon_{2}a^{2}c}{2}\k_{\bm{2}}^{2}+c_{3}\right)\N \\
		&\quad+\left(2\epsilon_{1}\epsilon_{2}c\k_{\bm{2}}^{(3)}+3\epsilon_{1}\epsilon_{2}ac\k_{\bm{1}}\k_{\bm{2}}'\right)\W_{\bm{2}}.
	\end{aligned}
\end{equation}

Note that the above geometric hierarchy of commuting vector fields at the curve level is a generalization of 
the ones obtained in \cite{del_amor_hamiltonian_2014} for the $3$-dimensional case, albeit using a different procedure. In fact, it was not possible to extend  the procedure used in \cite{del_amor_hamiltonian_2014} to obtain the recursion operator and the Hamiltonian structure at the curve level to the $4$-dimensional setting, mainly because of the appearance of nonlocal vector fields. Searching for a Hamiltonian structure at the curve level for the $4$-dimensional case will be one of the subject for future research.

\section{Conclusions}
In this paper, our primary aim was to study the integrability properties of null curve evolutions in a flat $4$-dimensional background. We undertook our research in an enough degree of generality for the purpose of showing the role of the constants appearing on it, especially when they possess geometrical meaning. In that regard, it is particularly important the way in which the computations were conducted to expose the most important elements of the Hamiltonian structure for curvature flows. One of the most surprising fact was to obtain the recursion operator (split into both the Poisson operator and the symplectic operator in formulas \eqref{split1} and \eqref{symplectic-cosymplectic}) of the Hirota-Satsuma system by means of the geometry of null curves  or, more precisely, making use of the projection of convenient variation vector fields onto the screen bundle. Similar results were obtain in \cite{beffa_integrable_2002,sanders_integrable_2003}, this suggesting that the screen bundle of a null curve may be thought of as playing the same role of the normal bundle in a Riemannian curve.  We can therefore also state the following important conclusion: if we have a evolution vector field $\V$ with $(\bm{\varphi}_{\V},\bm{\psi}_{\V})$ having the property of being the gradient of a certain functional $\bm{H}$, then the flow associated to the curvatures $(\V(\k_{\bm{1}}),\V(\k_{\bm{2}}))$ is a completely integrable Hamiltonian system. 

Furthermore, in Proposition \ref{prop-recursion} we have lifted the recursion operator for the Hirota-Satsuma system (at the curvature level) to a recursion operator for the NLIE equation (at the curve level), enabling us to obtain an infinite hierachy of commuting vector fields. Proposition \ref{prop-LieBracket2} shows that the subspace consisting of $\mathcal{A}$-local evolution vector fields (denoted by $\mathfrak{X}_{\mathcal{A}}(\Lambda)$) is closed under bracket and contains the commuting flows as a subalgebra. 

One of the many benefits of increasing the dimension of the ambient space has been that the connections between integrable hierarchies of both null curves and their curvature flows become clearer. Nevertheless, finding a Hamiltonian structure at the curve level still needs to be achieved. In addition, it would be interesting to develop a purely geometric method to construct the existing structures of the dynamic of null curve motions without lifting any element from the curvature flow.
Accordingly, further work is needed, perhaps in a nonlocal background, if possible, to properly understand which also has appeared in different contexts. 

\section*{Acknowledgments}
This work has been partially supported by MINECO (Ministerio de Economía y Competitividad) project MTM2015-65430-P, and by Fundación Séneca project 19901/GERM/15, Spain.


\begin{thebibliography}{10}

\bibitem{del_amor_hamiltonian_2014}
J.~{Del Amor}, {\'A}.~Gim{\'e}nez, and P.~Lucas.
\newblock {Hamiltonian structure for null curve evolution}.
\newblock {\em Nonlinearity}, 27(11):2627--2641, 2014.

\bibitem{musso_hamiltonian_2010}
E.~Musso and L.~Nicolodi.
\newblock {Hamiltonian flows on null curves}.
\newblock {\em Nonlinearity}, 23(9):2117{\textendash}2129, 2010.

\bibitem{li_motions_2013}
Y-Y Li.
\newblock {Motion of Cartan curves in $n$-dimensional Minkowski space}.
\newblock {\em Modern Physics Letters A}, 28(24):1--13, 2013.

\bibitem{hirota_soliton_1981}
R.~Hirota and J.~Satsuma.
\newblock {Soliton solutions of a coupled {Korteweg}-de {Vries} equation}.
\newblock {\em Physics Letters A}, 85(8-9):407--408, 1981.

\bibitem{dodd_integrability_1982}
R.~Dodd and A.~Fordy.
\newblock {On the integrability of a system of coupled {KdV} equations}.
\newblock {\em Physics Letters A}, 89(4):168--170, 1982.

\bibitem{weiss_modified_1985}
J.~Weiss.
\newblock {Modified equations, rational solutions, and the {Painlev{\'e}}
  property for the {Kadomtsev}-{Petviashvili} and {Hirota}-{Satsuma}
  equations}.
\newblock {\em Journal of Mathematical Physics}, 26(9):2174--2180, 1985.

\bibitem{weiss_sine-gordon_1984}
J.~Weiss.
\newblock {The sine-{Gordon} equations: {Complete} and partial integrability}.
\newblock {\em Journal of Mathematical Physics}, 25(7):2226--2235, 1984.

\bibitem{levi_hierarchy_1983}
D.~Levi.
\newblock {A hierarchy of coupled {Korteweg}-de {Vries} equations}.
\newblock {\em Physics Letters A}, 95(1):7--10, 1983.

\bibitem{leble_darboux_1993}
S.B. Leble and N.V. Ustinov.
\newblock {Darboux transforms, deep reductions and solitons}.
\newblock {\em Journal of Physics A: Mathematical and General},
  26(19):5007--5016, 1993.

\bibitem{hu_new_2003}
H.C. Hu and Q.P. Liu.
\newblock {New {Darboux} transformation for {Hirota}-{Satsuma} coupled {KdV}
  system}.
\newblock {\em Chaos, Solitons and Fractals}, 17(5):921--928, 2003.

\bibitem{hu_new_2008}
H.C. Hu and Y.~Liu.
\newblock {New positon, negaton and complexiton solutions for the
  {Hirota}-{Satsuma} coupled {KdV} system}.
\newblock {\em Physics Letters, Section A: General, Atomic and Solid State
  Physics}, 372(36):5795--5798, 2008.

\bibitem{fuchssteiner_lie_1982}
B.~Fuchssteiner.
\newblock {The Lie algebras structure of degenerate Hamiltonian and
  bi-Hamiltonian systems}.
\newblock {\em Progress of Theoretical Physics}, 68(4):1082--1104, 1982.

\bibitem{oevel_integrability_1983}
W.~Oevel.
\newblock {On the integrability of the {Hirota}-{Satsuma} system}.
\newblock {\em Physics Letters A}, 94(9):404--407, 1983.

\bibitem{nersessian_massive_1998}
A.~Nersessian and E.~Ramos.
\newblock {Massive spinning particles and the geometry of null curves}.
\newblock {\em Physics Letters, Section B: Nuclear, Elementary Particle and
  High-Energy Physics}, 445(1-2):123--128, 1998.

\bibitem{nersessian_particle_2000}
A.~Nersessian, R.~Manvelyan, and H.J.W. M{\"u}ller-Kirsten.
\newblock {Particle with torsion on 3d null-curves}.
\newblock {\em Nuclear Physics B - Proceedings Supplements}, 88(1-3):381--384,
  2000.

\bibitem{ferrandez_geometrical_2002}
A.~Ferr{\'a}ndez, A.~Gim{\'e}nez, and P.~Lucas.
\newblock {Geometrical particle models on {3D} null curves}.
\newblock {\em Physics Letters B}, 543(3-4):311{\textendash}317, 2002.

\bibitem{ferrandez_relativistic_2007}
A.~Ferr{\'a}ndez, A.~Gim{\'e}nez, and P.~Lucas.
\newblock {Relativistic particles and the geometry of {4-D} null curves}.
\newblock {\em Journal of Geometry and Physics}, 57(10):2124{\textendash}2135,
  2007.

\bibitem{gimenez_relativistic_2010}
A.~Gim{\'e}nez.
\newblock {Relativistic particles along null curves in {3D} Lorentzian space
  forms}.
\newblock {\em International Journal of Bifurcation and Chaos in Applied
  Sciences and Engineering}, 20(9):2851{\textendash}2859, 2010.

\bibitem{beffa_integrable_2002}
G.~{Mar{\'i} Beffa}, J.A. Sanders, and J.P. Wang.
\newblock {Integrable systems in three-dimensional Riemannian geometry}.
\newblock {\em Journal of Nonlinear Science}, 12(2):143{\textendash}167, 2002.

\bibitem{sanders_integrable_2003}
J.A. Sanders and J.P. Wang.
\newblock {Integrable systems in $n$-dimensional Riemannian geometry}.
\newblock {\em Moscow Mathematical Journal}, 3(4):1369--1393, 2003.

\bibitem{langer_poisson_1991}
J.~Langer and R.~Perline.
\newblock {Poisson geometry of the filament equation}.
\newblock {\em Journal of Nonlinear Science}, 1(1):71{\textendash}93, 1991.

\bibitem{mansfield_evolution_2006}
Elizabeth~L. Mansfield and Peter~H. van~der Kamp.
\newblock {Evolution of curvature invariants and lifting integrability}.
\newblock {\em Journal of Geometry and Physics}, 56(8):1294--1325, August 2006.

\bibitem{del_amor_lie_2015}
J.~{Del Amor}, A.~Gim{\'e}nez, and P.~Lucas.
\newblock {A Lie algebra structure on variation vector fields along curves in
  2-dimensional space forms}.
\newblock {\em Journal of Geometry and Physics}, 88:94--104, 2015.

\bibitem{dorfman_dirac_1993}
I.~Dorfman.
\newblock {\em {Dirac {Structures} and {Integrability} of {Nonlinear}
  {Evolution} {Equations}}}.
\newblock Wiley, Chichester, England; New York, 1 edition edition, June 1993.

\bibitem{dickey_soliton_2003}
L.A. Dickey.
\newblock {\em {Soliton Equations and Hamiltonian Systems}}, volume~26 of {\em
  {Advanced Series in Mathematical Physics}}.
\newblock World Scientific Publishing Co. Inc., River Edge, {NJ}, second
  edition, 2003.

\bibitem{blaszak_multi-hamiltonian_2012}
M.~Blaszak.
\newblock {\em {Multi-{Hamiltonian} {Theory} of {Dynamical} {Systems}}}.
\newblock Springer Science \& Business Media, December 2012.

\bibitem{mikhailov_symmetry_1987}
A.V. Mikhailov, A.B. Shabat, and R.I. Yamilov.
\newblock {The symmetry approach to the classification of non-linear equations.
  Complete lists of integrable systems}.
\newblock {\em Russian Mathematical Surveys}, 42(4):1--63, August 1987.

\bibitem{sokolov_symmetries_1988}
V.V. Sokolov.
\newblock {On the symmetries of evolution equations}.
\newblock {\em Russian Mathematical Surveys}, 43(5):165--204, October 1988.

\bibitem{mikhailov_symmetry_1991}
A.V. Mikhailov, A.B. Shabat, and V.V. Sokolov.
\newblock {The symmetry approach to classification of integrable equations}.
\newblock In Professor Dr Vladimir~E. Zakharov, editor, {\em {What Is
  Integrability?}}, {Springer Series in Nonlinear Dynamics}, pages 115--184.
  Springer Berlin Heidelberg, January 1991.

\bibitem{mikhailov_towards_1998}
A.V. Mikhailov and R.I. Yamilov.
\newblock {Towards classification of $(2+1)$-dimensional integrable equations:
  integrability conditions I}.
\newblock {\em Journal of Physics A: Mathematical and General}, 31(31):6707,
  August 1998.

\bibitem{ferrandez_null_2001}
A.~Ferr{\'a}ndez, A.~Gim{\'e}nez, and P.~Lucas.
\newblock {Null helices in Lorentzian space forms}.
\newblock {\em International Journal of Modern Physics A},
  16(30):4845{\textendash}4863, 2001.

\bibitem{yasui_differential_1998}
Y.~Yasui and N.~Sasaki.
\newblock {Differential geometry of the vortex filament equation}.
\newblock {\em Journal of Geometry and Physics}, 28(1--2):195--207, November
  1998.

\end{thebibliography}

\end{document}